\documentclass[a4paper,UKenglish,cleveref,autoref]{lipics-v2019}

\title{Timed Basic Parallel Processes}

\author{Lorenzo~Clemente}{University of Warsaw, Poland}{}{https://orcid.org/0000-0003-0578-9103}{Partially supported by Polish NCN grant 2017/26/D/ST6/00201.}

\author{Piotr~Hofman}{University of Warsaw, Poland}{}{https://orcid.org/0000-0001-9866-3723}{Partially supported by Polish NCN grant 2016/21/D/ST6/01368.}	

\author{Patrick~Totzke}{University of Liverpool, UK}{}{https://orcid.org/0000-0001-5274-8190}{}

\authorrunning{L.~Clemente and P.~Hofman and P.~Totzke}

\Copyright{Lorenzo Clemente and Piotr Hofman and Patrick Totzke}

\ccsdesc[500]{Theory of computation~Timed and hybrid models}

\keywords{Timed Automata, Petri Nets}

\category{}

\relatedversion{This is the full version of a paper presented at
    the
    30th International Conference on Concurrency Theory (CONCUR 2019).
}

\supplement{}

\acknowledgements{Many thanks to Rasmus Ibsen-Jensen for helpful discussions and pointing us towards \cite{HIM2013}.}

\nolinenumbers 

\hideLIPIcs

\usepackage{amsfonts,amstext,amsmath,amssymb,mathrsfs,calc}
\usepackage{bm}
\usepackage{makecell}

\usepackage[table]{xcolor}

\usepackage{todonotes}

\usepackage{mathtools}
\usepackage{xspace}
\usepackage{extarrows}
\usepackage{pdfsync}
\usepackage{stmaryrd}
\usepackage{macros}
\usepackage{thm-restate}

\newcommand{\Aut}{{\?A}}
\newcommand{\Baut}{{\?B}}

\begin{document}
\maketitle
\begin{abstract}

Timed basic parallel processes (\TBPP) extend communication-free Petri nets (aka.~\BPP or commutative context-free grammars) by a global notion of time.
\TBPP can be seen as an extension of timed automata (\TA) with context-free branching rules,
and as such may be used to model networks of independent timed automata with process creation.

We show that the coverability and reachability problems (with unary encoded target multiplicities)
are \PSPACE-complete and \EXPTIME-complete, respectively. 
For the special case of 1-clock \TBPP, both are \NP-complete and hence not more complex than for untimed \BPP.
This contrasts with known super-Ackermannian-completeness and undecidability results for general timed Petri nets. 

As a result of independent interest, and basis for our \NP upper bounds,
we show that the reachability relation of 1-clock \TA
can be expressed by a formula of polynomial size in the existential fragment of linear arithmetic,
which improves on recent results from the literature.

\end{abstract}

\section{Introduction}
\label{sec:intro}

We study safety properties of unbounded networks of timed processes,
where time is global and elapses at the same rate for every process.
Each process is a \emph{timed automaton} (\TA) \cite{AD1994} controlling its own set of private clocks,
not accessible to the other processes.
A process can dynamically create new sub-processes, which are thereafter independent from each other and their parent,
and can also terminate its execution and disappear from the network.

While such systems 
can be conveniently modelled in \emph{timed Petri nets} (\TdPN), 
verification problems for this model are either undecidable or prohibitively complex:
The reachability problem is undecidable even when individual processes carry only one clock \cite{RFE1999}
and the coverability problem is undecidable for two or more clocks.
In the one-clock case coverability remains decidable but
its complexity is hyper-Ackermannian \cite{AACK2001,HSS2012}.

These hardness results however require unrestricted synchronization between processes,
which motivates us to study of the communication-free fragment of \TdPN,
called \emph{timed basic parallel processes} (\TBPP) in this paper.
This model subsumes both \TA and communication-free Petri nets (a.k.a.~\BPP \cite{C1993,E1997}).
The general picture that we obtain is that extending communication-free Petri nets by a global notion of time
comes at no extra cost in the complexity of safety checking, and it improves on the prohibitive complexities of \TdPN.

\mypar{Our contributions.}
We show that the \TBPP coverability problem is \PSPACE-complete,
matching same complexity for \TA \cite{AD1994, FJ2015},
and that the more general \TBPP reachability problem is \EXPTIME-complete,
thus improving on the undecidability of \TdPN.
The lower bounds already hold for \TBPP with two clocks if constants are encoded in binary;
\EXPTIME-hardness for reachability with no restriction on the number of clocks holds for constants in $\set{0,1}$.
The upper bounds are obtained by reduction to \TA reachability and reachability games \cite{JT2007},
and assume that process multiplicities in target configurations are given in unary.

In the single-clock case, we show that both \TBPP coverability and reachability are \NP-complete,
matching the same complexity for (untimed) \BPP \cite{E1997}.
This paves the way for the automatic verification of unbounded networks of 1-clock timed processes,
which is currently lacking in mainstream verification tools such as UPPAAL \cite{LPY1997} and KRONOS \cite{Yovine:KRONOS:1997}.
The \NP lower bound already holds when the target configuration has size $2$;
when it has size one, 1-clock \TBPP coverability becomes \NL-complete,
again matching the same complexity for 1-clock \TA \cite{LMS2004}
(and we conjecture that 1-clock reachability is in \PTIME under the same restriction).

As a contribution of independent interest, we show that the ternary reachability relation of 1-clock \TA
can be expressed by a formula of existential linear arithmetic (\ELA) of polynomial size.
By \emph{ternary reachability relation} we mean the family of relations $\set{\tareach {pq} {}}$
s.t.~$\mu \tareach {pq} \delta \nu$ holds if from control location $p$ and clock valuation $\mu \in \nnreals^k$
it is possible to reach control location $q$ and clock valuation $\nu \in \nnreals^k$
in exactly $\delta \in \nnreals$ time.
This should be contrasted with analogous results (cf.~\cite{FQSW:ArXiv:2019}) which construct formulas of exponential size,
even in the case of 1-clock \TA.
Since the satisfiability problem for \ELA is decidable in \NP,
we obtain a \NP upper bound to decide ternary reachability $\tareach {pq} {}$.
We show that the logical approach is optimal by providing a matching \NP lower bound for the same problem.
Our \NP upper bounds for the 1-clock \TBPP coverability and reachability problems
are obtained as an application of our logical expressibility result above, and the fact that \ELA is in \NP;
as a further technical ingredient we use polynomial bounds on the piecewise-linear description of value functions in 1-clock priced timed games \cite{HIM2013}.

\mypar{Related research.}
Starting from the seminal \PSPACE-completeness result of the nonemptiness problem for \TA \cite{AD1994} (cf.~also \cite{FJ2015}),
a rich literature has emerged considered more challenging verification problems,
including the symbolic description of the reachability relation
\cite{CJ1999, D:LICS:2002, KP:FSTTCS:2005, Dima:VISSAS:2005, FQSW:ArXiv:2019}.
There are many natural generalizations of \TA to add extra modelling capabilities,
including \emph{time Petri Nets} \cite{M1974, LP1991} (which associate timing constraints to transitions)
the already mentioned timed Petri nets (\TdPN) \cite{RFE1999,AACK2001,HSS2012} (where tokens carry clocks which are tested by transitions),
\emph{networks of timed processes} \cite{AbdullaJonsson:TCS03},
several variants of \emph{timed pushdown automata} \cite{BER1994,Dang:CAV:2001,BenerecettiMinopoliPeron:RTA:2010,TrivediWojtczak:RTA:2010,AAS:LICS:2012,Quaas:LMCS:2015,BenerecettiPeron:TCS:2016,ClementeLasotaLazicMazowiecki:LICS:2017,ClementeLasota:ICALP:2018},
\emph{timed communicating automata} \cite{KrcalYi:CTA:CAV:2006, ClementeHerbreteauStainerSutre:FOSSACS13, AbdullaAtigKrishna:TCA:FORMATS:2018, Clemente:TCA:ArXiv:2018},
and their lossy variant \cite{AbdullaAtigCederberg:TLCS:FSTTCS12},
and timed process calculi based on Milners CCS (e.g.~\cite{BLS2000}).
While decision problems for \TdPN have prohibitive complexity/are undecidable,
it has recently been shown that structural safety properties are \PSPACE-complete using forward accelerations \cite{AACMT2018}.

\mypar{Outline.}
In \cref{sec:prep} we define \TBPP and their reachability and coverability decision problems.
In \cref{sec:ha} we show that the reachability relation for 1-clock timed automata can be expressed in polynomial time in an existential formula of linear arithmetic,
and that the latter logic is in \NP.
We apply this result in Sec.~\ref{sec:oneclock} to show that the reachability and coverability problems for 1-clock \TBPP is \NP-complete.
Finally, in \cref{sec:multiclock} we study the case of \TBPP with $k \geq 2$ clocks,
and in \cref{sec:conclusion} we draw conclusions.

\section{Preliminaries}
\label{sec:prep}
\mypar{Notations.}

We use $\nat$ and $\nnreals$ to denote the sets
of nonnegative integers and reals, respectively. 
For $c\in\nnreals$ we write
$\intof{c} \in \nat$ for its integer part
and $\fractof{c} \eqdef c- \intof c$
for its fractional part.
For a set $\?X$, we use $\?X^*$ to denote the set of finite sequences over $\?X$
and $\msets{\?X}$ to denote the set of finite multisets over $\?X$,
i.e., functions $\N^{\?X}$.
We denote the empty multiset by $\zerro$,
we denote the union of two multisets $\alpha, \beta \in \N^{\?X}$ by $\alpha \msplus \beta$,
which is defined point-wise,
and by $\alpha \leq \beta$ we denote the natural partial order on multisets,
also defined point-wise.
The \emph{size} of a multiset $\alpha \in \N^{\?X}$ is $\size \alpha = \sum_{X \in {\?X}} \alpha(X)$.
We overload notation and we let $X \in \?X$ denote the singleton multiset of size $1$ containing element $X$.
For example, if $X, Y \in \?X$, then the multiset consisting of 1 occurrence of $X$ and 2 of $Y$ will be denoted by $\alpha = X \msplus Y \msplus Y$;
it has size $\size \alpha = 3$.

\mypar{Clocks.}
Let $\clockset$ be a finite set of \emph{clocks}.
A \emph{clock valuation} is a function $\mu \in \nnreals^\clockset$ assigning a nonnegative real to every clock.
For $t\in\nnreals$, we write $\mu+t$ for the valuation that maps clock $x \in \clockset$ to $\mu(x)+t$. 
For a clock $x \in \clockset$ and a clock or constant $e \in \clockset \cup \N$
let $\mu[x := e]$ be the valuation $\nu$ s.t.~$\nu(x) = \mu(e)$ and $\nu(z) = \mu(z)$ for every other clock $z \neq x$
(where we assume $\mu(k) = k$ for a constant $k \in \N$);
for a sequence of assignments $R = (x_1 := e_1; \cdots; x_n := e_n)$
let $\mu[R] = \mu[x_1 := e_1]\cdots[x_n := e_n]$.
A \emph{clock constraint} is a conjunction of linear inequalities of the form
$c\bowtie k$, where $c\in\clockset$, $k\in\N$, and $\bowtie\; \in\set{<, \le, =, \ge, >}$;
we also allow $\true$ for the trivial constraint which is always satisfied.
We write 
$\mu\models\varphi$ to denote that the valuation $\mu$ satisfies the constraint $\varphi$.

\mypar{Timed basic parallel processes.}
    A \emph{timed basic parallel process} (\TBPP)
    consists of finite sets 
    $\clockset$,
    $\?X$, and 
    $\?R$
    of clocks,
    nonterminal symbols, 
    and rules. Each rule is of the form
    $$ X\Rule{\varphi; R}\alpha $$
    where $X\in\?X$ is a nonterminal, $\varphi$ is a clock constraint,
    $R$ is a sequence of assignments of the form $x := e$, where $e$ is either a constant in $\N$ or a clock in $\?C$,
    and $\alpha\in \?X^\oplus$ is a finite multiset of successor nonterminals
    \footnote{\label{footnote:sugar}We note that clock updates $x := k$ with $k \in \N$ can be encoded with only a polynomial blow-up
    by replacing them with $x := 0$,
    while recording in the finite control the last update $k$,
    and replacing a test $x \bowtie h$ with $x \bowtie h - k$.
    We use them as a syntactic sugar to simplify the presentation of some constructions.}.
    Whenever the test $\varphi \equiv \true$ is trivial,
    or $R$ is the empty sequence,
    we just omit the corresponding component and just write
    $X\Rule{\varphi}\alpha$, $X\Rule{R}\alpha$, or $X\Rule{}\alpha$.
    Finally, we say that we \emph{reset} the clock $x_i$ if we assign it to $0$.

    Henceforth, we assume w.l.o.g.~that the size $\len{\alpha}$ is at most 2.
    A rule with $\alpha=\zerro$ is called a \emph{vanishing} rule,
    and a rule with $\len\alpha = 2$ is called a \emph{branching} rule.
    We will write $k$-\TBPP to denote the class of \TBPP with $k$ clocks.
    
    \medskip
    A \emph{process} is a pair $(X, \mu) \in \?X\x\nnreals^\clockset$ comprised of a nonterminal $X$ and a clock valuation $\mu$,
    and a \emph{configuration} $\alpha$ is a multiset of processes, i.e., $\alpha \in \msets{(\?X\x\nnreals^\clockset)}$.
    For a process $P = (X, \mu)$ and $t \in \nnreals$, we denote by $P + t$ the process $(X, \mu + t)$,
    and for a configuration $\alpha = P_1 \msplus \cdots \msplus P_n$,
    we denote by $\alpha + t$ the configuration $Q_1 \msplus \cdots \msplus Q_n$, where $Q_1 = P_1 + t, \dots, Q_n = P_n + t$.
    The semantics of a \TBPP
    $(\clockset,\?X,\?R)$
    is given by an infinite timed transition system $(C, \to)$,
    where $C = \msets{(\?X\x\nnreals^\clockset)}$ is the set of configurations,
    and $\to \subseteq C \times \nnreals \times C$ is the transition relation between configurations.
    There are two kinds of transitions:
    
    \begin{description}
    \item   [Time elapse:]
            For every configuration $\alpha \in C$ and $t\in\nnreals$,
            there is a transition $\alpha\tstep{t}\alpha+t$
            in which all clocks in all processes are simultaneously increased by $t$.
            In particular, the empty configuration stutters: $\emptyset \tstep t \emptyset$, for every $t \in \nnreals$.

    \item   [Discrete transitions:]
            For every configuration $\gamma = \alpha \msplus (X, \mu) \msplus \beta \in C$
            and rule $X \Rule{\varphi; R} Y \msplus Z$ s.t.~$\mu \models\varphi$
            there is a transition $\gamma \step 0 \alpha \msplus (Y, \nu) \msplus (Z, \nu) \msplus \beta$,
            where $\nu = \mu[R]$. 
            Analogously, rules
            $X\Rule{\varphi; R}Y$
            and $X\Rule{\varphi; R}\zerro$
            induce transitions
            $\gamma \step 0 \alpha \msplus (Y, \nu) \msplus \beta$
            and 
            $\gamma \step 0 \alpha \msplus \beta$. 
    \end{description}
    A \emph{run} starting in $\alpha$ and ending in $\beta$
    is a sequence of transitions $\alpha = \alpha_0 \reach {t_1} \alpha_1 \cdots \reach {t_n} \alpha_n = \beta$.
    We write $\alpha\reach t \beta$ whenever there is a run as above where the sum of delays is $t = t_1 + \cdots + t_n$,
    and we write $\alpha\reach *\beta$ whenever $\alpha\reach t \beta$ for some $t \in \nnreals$.

    \medskip
    \noindent
    \TBPP generalise several known models: 
    A \emph{timed automaton} (\TA) \cite{AD1994}
    is a \TBPP without branching rules;
    in the context of \TA, we will sometimes call nonterminals with the more standard name of \emph{control locations}.
    Untimed \emph{basic parallel processes} (\BPP) \cite{C1993, E1997}
    are \TBPP over the empty set of clocks $\?X = \emptyset$.
    \TBPP can also be seen as a structural restriction of \emph{timed Petri nets} \cite{AACK2001, HSS2012}
    where each transition consumes only one token at a time. 

    \TBPP are related to \emph{alternating timed automata} (\ATA) \cite{OuaknineWorrell:MTL:LICS:2005, LasotaWalukiewicz:ATA:ACM08}:
    Branching in \TBPP rules corresponds to universal transitions in \ATA.
    However, \ATA offer additional means of synchronisation between the different branches of a run tree:
    While in a \TBPP synchronisation is possible only through the elapse of time,
    in an \ATA all branches must read the same timed input word.

\mypar{Decision problems.}

We are interested in checking safety properties of \TBPP in the form of the following decision problems.
The \emph{reachability problem} asks whether a target configuration is reachable from a source configuration.

\dproblem{
    A \TBPP $(\clockset,\?X,\?R)$, an initial $X \in \?X$
    and target nonterminals $T_1, \dots, T_n \in \?X$.
    }{ 
    Does $(X, \vec{0}) \reach{*} (T_1, \vec{0}) \msplus \cdots \msplus (T_n, \vec{0})$ hold?
}
\noindent 
It is crucial that we reach all processes in the target configurations \emph{at the same time},
which provides an external form of global synchronisation between processes.

Motivated both by complexity considerations and applications for safety checking,
we study the \emph{coverability problem},
where it suffices to reach some configuration larger than the given target in the multiset order.
For configurations $\alpha, \beta \in \msets{(\?X\x\nnreals^\clockset)}$,
let $\alpha \rcovers \beta$ whenever there exists $\gamma \in \msets{(\?X\x\nnreals^\clockset)}$
s.t.~$\alpha \reach{*} \gamma \geq \beta$.

\dproblem{
    A \TBPP $(\clockset,\?X,\?R)$, an initial $X \in \?X$
    and target nonterminals $T_1, \dots, T_n \in \?X$.
    }{ 
    Does $(X, \vec{0}) \rcovers (T_1, \vec{0}) \msplus \cdots \msplus (T_n, \vec{0})$?
}

The \emph{simple} reachability/coverability problems are as above
but with the restriction that the target configuration is of size $1$, i.e., a single process.
Notice that this is a proper restriction,
since reachability and coverability do not reduce in general to their simple variant.
Finally, the \emph{non-emptiness problem} is the special case of the reachability problem
where the target configuration $\alpha$ is the empty multiset $\emptyset$.

In all decision problems above the restriction to zero-valued clocks in the initial process is mere convenience,
since we could introduce a new initial nonterminal $Y$
and a transition $Y \Rule{x_1 := \mu(x_1); \dots; x_n := \mu(x_n)} X$
initialising the clocks to the initial values provided by the (rational) clock valuation $\mu \in \nnrationals^\clockset$.
Similarly, if we wanted to reach the final configuration $\alpha = (X_1, \mu_1) \msplus (X_2, \mu_2)$
with $\mu_1, \mu_2 \in \nnrationals^\clockset$,
then we could add two nonterminals $Y_1, Y_2$
and two new rules
$X_1 \Rule {x_1 = \mu_1(x_1) \wedge \cdots \wedge x_n = \mu_1(x_n); x_1:=0; \dots; x_n:=0 } Y_1$ and
$X_2 \Rule {x_1 = \mu_2(x_1) \wedge \cdots \wedge x_n = \mu_2(x_n); x_1:=0; \dots; x_n:=0 } Y_2$
and check whether $X \reach{*} (Y_1, \vec{0}) \msplus (Y_2, \vec{0})$ holds.
(It is standard to transform \TBPP with constraints of the form $x_i = k$ with $k \in \nnrationals$
in the form $x_i = k$ with $k \in \N$.)
Similarly, the restriction of having just one initial nonterminal process is also w.l.o.g.,
since if we wanted to check reachability from $(X_1, \vec{0}) \msplus (X_2, \vec{0})$
we could just add a new initial nonterminal $X$ and a branching rule $X \Rule {} X_1 \msplus X_2$.

For complexity considerations we will assume that all constants appearing in clock constraints are given in binary encoding, and that the multiplicities of target processes are in unary.

\section{Reachability Relations of One-Clock Timed Automata}
\label{sec:ha}
In this section we show that the reachability relation of 1-clock \TA
is expressible as an existential formula of linear arithmetic of polynomial size.
Since the latter fragment is in \NP,
this gives an \NP algorithm to check whether a family of \TA can reach the respective final locations \emph{at the same time}.
This result will be applied in Sec.~\ref{sec:oneclock} to show that coverability and reachability of 1-\TBPP are in \NP.
\ignore{
There have been several works studying the problem of computing \TA reachability relations
\cite{CJ1999, D:LICS:2002, KP:FSTTCS:2005, Dima:VISSAS:2005, QSW2017, FQSW:ArXiv:2019}.
In particular, \cite{FQSW:ArXiv:2019} shows that the reachability relation of a $k$-clock \TA
is computable in time polynomial in the size of the automaton
and exponential in the number of clocks $k$ and bit length of the maximum constraint constant.
We improve on this result by showing that, in the case of $k = 1$ clocks,
the formula is polynomial altogether (even in the bit length of the maximum constant).
}
We first show that existential linear arithmetic is in \NP
(which is an observation of independent interest),
and then how to express the reachability relation of 1-\TA in existential linear arithmetic in polynomial time.

The set of \emph{terms} $t$ is generated by the following abstract grammar
\begin{align*}
	s, t \; ::= \; x \sep k \sep \floor t \sep \fractof t \sep {-t} \sep s + t \sep k \cdot t,
\end{align*}
where $x$ is a rational variable,
$k \in \Z$ is an integer constant encoded in binary,
$\floor t$ represents the integral part of $t$,
and $\fractof t$ its fractional part.
\emph{Linear arithmetic} (\LA) is the first order language
with atomic proposition of the form $s \leq t$ \cite{W1999},
we denote by \ELA its existential fragment,
and by \QFLA its quantifier-free fragment.
Linear arithmetic generalises both \emph{Presburger arithmetic} (\PA) 
and \emph{rational arithmetic} (\RA), 
whose existential fragments are known to be in \NP~\cite{P1930,FR1975}.
This can be generalised to \ELA.
(The same result can be derived from the analysis of \cite[Theorem 3.1]{BJW:ACM:2005}).

\begin{restatable}{theorem}{logicNP}
	\label{thm:logic:NP}
	The existential fragment \ELA of \LA is in \NP.
\end{restatable}

\let\oldst=\st
\renewcommand{\st}{s.t.\xspace} 

\newcommand{\eqmod}[2]{\equiv {#2} \; (\!\!\!\!\!\!\mod {#1})}

\newcommand{\trans}{\textcolor{black}{r}}
\newcommand{\ttime}{\textcolor{black}{t}}

Let $\Aut = (\?C, \?X, \?R)$ be a $k$-\TA.
The \emph{ternary reachability relation} of $\Aut$
is the family of relations $\setsmall{\tareach {XY} {}}_{X, Y \in \?X}$,
where each $\tareach {XY} {} \subseteq \nnreals^\?C \x \nnreals \x \nnreals^\?C$
is defined as: $\mu \tareach {XY} \delta \nu$ iff $(X, \mu) \reach \delta (Y, \nu)$.
We say that the reachability relation is \emph{expressed} by a family of \LA formulas $\setsmall{\varphi_{XY}}_{X, Y \in \?X}$ if
\begin{align*}
	\mu \tareach {XY} \delta \nu \quad \textrm{ iff } \quad
		(\mu, \delta, \nu) \models \varphi_{XY}(\vec x, \ttime, \vec y),
		\textrm{ for every } X, Y \in \?X, \mu, \nu \in \nnreals^\?C, \delta \in \nnreals.
\end{align*}
In the formula $\varphi_{XY}(\vec x, \ttime, \vec y)$,
$\vec x$ are $k$ variables representing the clock values in location $X$ at the beginning of the run,
$\vec y$ are $k$ variables representing the clock values in location $Y$ at the end of the run,
and $\ttime$ is a single variable representing the total time elapsed during the run.
In the rest of this section, we assume that the \TA has only one clock $\?X = \set{x}$.

The main result of this section is that 1-\TA reachability relations are expressible by \ELA formulas constructible in polynomial time.

\begin{theorem}
	\label{thm:1-clock:TA}
	Let $\Aut$ be a 1-\TA.
	The reachability relation $\setsmall{\tareach {XY} {}}_{X, Y \in \?X}$
	is expressible as a family of formulas $\setsmall{\varphi_{XY}}_{X, Y \in \?X}$ of existential linear arithmetic \ELA
	in polynomial time.
\end{theorem}

\noindent
In the rest of the section we prove the theorem above.
We begin with some preliminaries.

\mypar{Interval abstraction.}
We replace the integer value of the clock $x$ by its interval \cite{LMS2004}.
Let $0 = k_0 < k_1 < \cdots < k_n < k_{n+1} = \infty$
be all integer constants appearing in constraints of $\Aut$,
and let the set of intervals be the following totally ordered set:
\begin{align*} 
	\Lambda = \set{ \set {k_0} < (k_0, k_1) < \set {k_1} < \cdots < (k_{n-1}, k_n) < \set {k_n} < (k_n, k_{n+1})}.
\end{align*}
Clearly, we can resolve any constraint of $\Aut$
by looking at the interval $\lambda \in \Lambda$.
We write $\lambda \models \varphi$ whenever $v \models \varphi$ for some $v \in \lambda$
(whose choice does not matter by the definition of $\lambda$).

\ignore{
\mypar{Fractional region abstraction.}
We introduce an extra clock $x_0$ which is $0$ at the beginning, is never reset, and never appears in any constraint.
It is used only to keep track of the fractional elapse of time,
and its relationship with the fractional part of clock $x$.
We distinguish the following 6 fractional regions relating the fractional parts of $x_0$ and $x$:
\begin{align*}
	\xi_0(x_0, x) \;&\equiv\; 0 = \fractof {x_0} = \fractof x < 1
	& \xi_1(x_0, x) \;&\equiv\; 0 = \fractof {x_0} < \fractof x < 1 \\
	\xi_2(x_0, x) \;&\equiv\; 0 < \fractof {x_0} = \fractof x < 1
	&\xi_3(x_0, x) \;&\equiv\; 0 < \fractof {x_0} < \fractof x < 1 \\
	\xi_4(x_0, x) \;&\equiv\; 0 = \fractof {x} < \fractof {x_0} < 1
	&\xi_5(x_0, x) \;&\equiv\; 0 < \fractof {x} < \fractof {x_0} < 1.
\end{align*}
Let $\Xi = \set {\xi_0, \dots, \xi_5}$.
When time elapses, and when $x$ is reset,
we move deterministically from region to region according to the following rules:
\begin{align*}
	&\xi_0 \goesto {\elapse} \xi_2 \goesto {\elapse} \xi_0
	&&\xi_1 \goesto {\elapse} \xi_3 \goesto {\elapse} \xi_4 \goesto {\elapse} \xi_5 \goesto {\elapse} \xi_1 \\
	&\resetop {\xi_0} = \resetop {\xi_1} = \xi_0
	&& \resetop {\xi_2} = \resetop {\xi_3} = \resetop {\xi_4} = \resetop {\xi_5} = \xi_4.
\end{align*}
}

\mypar{The construction.}
Let $\Aut = (\set{x}, \?X, \?R)$ be a \TA.
In order to simplify the presentation below,
we assume w.l.o.g.~that the only clock updates are resets $x := 0$ (cf.~footnote \ref{footnote:sugar}).
We build an \NFA $\Baut = (\Sigma, Q, \to)$
where $\Sigma$ contains symbols $(r, \varepsilon)$ and $(r, \tick \lambda)$
for every transition $r \in \?R$ of $\Aut$ and interval $\lambda \in \Lambda$,
and an additional symbol $\tau$ representing time elapse,
and $Q = \?X \times \Lambda$ 
is a set of states of the form $(X, \lambda)$,
where $X \in \?X$ is a control location of $\Aut$
and $\lambda \in \Lambda $ is an interval.
Transitions $\to \subseteq Q \x \Sigma \x Q$ are defined as follows.
A rule $r = X \Rule {\varphi; R} Y \in \?R$ of $\Aut$ generates one or more transitions in $\Baut$
of the form $$(X, \lambda) \goesto {(r, a)} (Y, \mu)$$
whenever $\lambda\models \varphi$ and any of the following two conditions is satisfied:
\begin{itemize}
	\item the clock is not reset i.e $R$ is equal $x:=x$, and $\mu = \lambda, a = \varepsilon$, or
	\item the clock is reset $x:= 0$, $\mu = \set 0$, 
	and the automaton emits a tick $a = \tick \lambda$.
\end{itemize}
A time elapse transition is simulated in $\Baut$ by transitions of the form
\begin{align*}
	(X, \lambda) &\goesto \tau (X, \mu), \qquad
	\lambda \leq \mu \textrm{ (the total ordering on intervals)}. 
\end{align*}

\mypar{Reachability relation of $\Aut$.}

For a set of finite words $L \subseteq \Sigma^*$,
let $\psi_L(\vec y)$ be a formula of existential Presburger arithmetic 
with a free integral variable $y_\lambda$ for every interval $\lambda \in \Lambda$
counting the number of symbols of the form $(\trans, \tick \lambda)$, for some $\trans \in \?R$.
The formula $\psi_L$ can be computed from the Parikh image of $L$:
By~{\cite[Theorem 4]{VSS05}}, a formula $\tilde\psi_L(\vec z)$ of existential Presburger arithmetic
can be computed in linear time from an \NFA (or even a context-free grammar) recognising $L$,
and then one just defines
$\psi_L(\vec y) \equiv \exists \vec z \cdot \tilde\psi_L(\vec z) \wedge \bigwedge_{\lambda \in \Lambda} y_\lambda = \sum_{\trans \in \?R} z_{\trans, \lambda}$.
Let $L_{cd}$ be the regular language recognised by $\Baut$ by making $c$ initial and $d$ final,
and let $\Lambda = \set{\lambda_0, \dots, \lambda_{2n+1}}$ contain $2n+2$ intervals.
Let $\psi_{cd}(x, \ttime, x')$ 
be a formula of existential Presburger arithmetic computing the total elapsed time $\ttime$,
given the initial $x$ and final $x'$ values of the unique clock:
\begin{align*}
	\psi_{cd} (x, \ttime, x') \; \equiv \;
		&\exists y_0, \dots, y_{2n+1} \cdot \psi_{L_{cd}}(\floor {y_0}, \dots, \floor {y_{2n+1}}) \ \wedge \ \\
		&\exists z_0, \dots, z_{2n+1} \cdot \bigwedge_{\lambda_i \in \Lambda} (z_i \in y_i \cdot \lambda_i) \ \wedge \ 
		\ttime = x' - x + \sum_{\lambda_i \in \Lambda} z_i, \textrm{ where } \\
	z \in y \cdot \lambda \; \equiv \; &\left\{
	\begin{array}{ll}
		a \cdot y < z < b \cdot y & \textrm{ if } \lambda = (a, b), \\
		z = a \cdot y & \textrm{ if } \lambda = \set {a}.
	\end{array}\right.
\end{align*}
Intuitively, $y_i$ represents the total number of times the clock is reset while in interval $\lambda_i$,
and $z_i$ represents the sum of the values of the clock when it is reset in interval $\lambda_i$.
For control locations $X, Y$ of $\Aut$, let
\begin{align*}
	\varphi_{XY}(x, \ttime, x') \; &\equiv \;
		\bigvee_{\lambda, \mu \in \Lambda} \setof {x \in \lambda \wedge x' \in \mu \wedge \psi_{cd}(x, \ttime, x')}
			{c = (X, \lambda),
			d = (Y, \mu)}.
\end{align*}

\noindent
The correctness of the construction is stated below.

\begin{restatable}{lemma}{TAreachLemma}
	For every configurations $(X, u)$ and $(Y, v)$ of $\Aut$
	and total time elapse $\delta \geq 0$,
	\begin{align*}
		u \tareach {XY} \delta v \quad \textrm{ iff } \quad (u, \delta, v) \models \varphi_{XY} (x, \ttime, x').
	\end{align*}
\end{restatable}

We conclude this section by applying \cref{thm:1-clock:TA} to solve the 1-\TA ternary reachability problem.
The \emph{ternary reachability problem} takes as input a \TA $\Aut$ as above,
with two distinguished control locations $X, Y \in \?X$,
and a total duration $\delta \in \Q$ (encoded in binary),
and asks whether $(\vec{0}) \tareach {XY} \delta (\vec{0})$.
The result below shows that computing 1-\TA reachability relations is optimal in order to solve the ternary reachability problem.
\begin{theorem}
    \label{thm:1-ta-ternary-reachability}
    The ternary reachability problem for 1-\TA is \NP-complete.
\end{theorem}

\begin{proof}
	For the upper bound, apply \cref{thm:1-clock:TA} to construct in polynomial time a formula of \ELA 
	expressing the reachability relation and check satisfiability in \NP thanks to \cref{thm:logic:NP}.
	
	The lower bound can be seen by reduction from \textsc{SubsetSum}.
	Let $\?S = \set{a_1, \dots, a_k} \subseteq \N$ and $a \in \N$ be the input to the subset sum problem,
	whereby we look for a subset $\?S' \subseteq \?S$ s.t.~$a = \sum_{b \in \?S'} b$.
	We construct a \TA with a single clock $x$ and locations $\?X = \set{X_0, \dots, X_{k}}$,
	where $X_0$ is the initial location and $X_k$ the target.
	A path through the system describes a subset by spending exactly $0$ or $a_i$ time in location $X_i$
	(see \cref{fig:1clock-ta-np}).
	In the constructed automaton,
	$(X_1,0)\reach {a}(X_k,0)$ iff the subset sum instance was positive.
	\begin{figure}[htpb]
	\centering
	\includegraphics[width=0.8\linewidth]{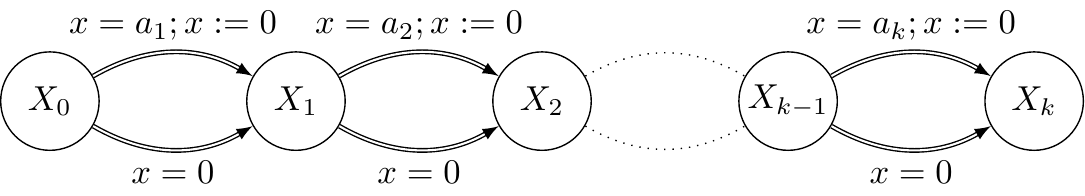}
	\caption{Reduction from subset sum to 1-\TA (ternary) reachability.
		We have $(X_0, 0)\reach{t} (X_k, 0)$ iff
		$t = \sum_{b \in \?S'} b$
		for some subset $\?S' \subseteq \set{a_1, \dots, a_k}$.
	}
	\label{fig:1clock-ta-np}
	\end{figure}
\end{proof}

\let\st=\oldst

\section{One-Clock \TBPP}
\label{sec:oneclock}
As a warm-up we note that the simple coverability problem for $1$-\TBPP,
where the target has size one, is inter-reducible with the reachability problem for $1$-clock \TA and hence \NL-complete \cite{LMS2004}.

\begin{theorem}
    \label{thm:simple-cover-1clock}
    The simple coverability problem for 1-clock \TBPP is \NL-complete.
\end{theorem}

\begin{proof}
    The lower bound is trivial since $1$-\TBPP generalize $1$-\TA.
    For the other direction we can transform a given \TBPP
    into a \TA by replacing branching rules of the form $X \Rule{\varphi,R} Y \msplus Z$
    with two rules $X \Rule{\varphi, R} Y$ and $X \Rule{\varphi, R} Z$.
    In the constructed \TA we have
    $(X,\mu)\reach{*}(Y,\nu)$ if, and only if,
    $(X,\mu)\reach{*}(Y,\nu)\msplus \gamma$ for some $\gamma$ in the original \TBPP.
\end{proof}
The construction above works because the target is a single process and so there are no constraints on the other processes in $\gamma$, which were produced as side-effects by the branching rules.
The (non-simple) 1-\TBPP coverability problem is in fact \NP-complete.
Indeed, even for untimed \BPP, coverability is \NP-hard,
and this holds already when target sets are encoded in unary (which is the setting we are considering here) \cite{E1997}. 
We show that for 1-\TBPP this lower bound already holds if the target has fixed size $2$.

\begin{lemma}
    \label{lem:cover-1clock-hard}
    Coverability is \NP-hard for 1-\TBPP
    already for target sets of size $\geq 2$.
\end{lemma}

\newcommand{\setA}{\mathbb{A}}
\newcommand{\setB}{\mathbb{B}}

\begin{proof}
    We proceed by reduction from subset sum as in \cref{thm:1-ta-ternary-reachability}.
    The only difference will be that here, we use one extra process to keep track of the total elapsed time.
    Let $\?S = \set{a_1, \dots, a_k} \subseteq \N$ and $t \in \N$ be the input to the subset sum problem.
    We construct a 1-\TBPP with nonterminals $\?X = \set{S, X_0, \dots, X_k, Y}$,
    where $S$ is the initial nonterminal and $Y$ will be used to keep track of the total time elapsed.
    The rules are as in the proof of \cref{thm:1-ta-ternary-reachability},
    and additionally we have an initial branching rule $S \Rule{x=0} X_0 \msplus Y$.
    We have $(S, 0) \rcovers(X_k, 0)\msplus (Y, t)$ if, and only if,
    $t = \sum_{b \in \setB} b$ for some subset $\?S'\subseteq \?S$.
\end{proof}

In the remainder of this section we will argue (\cref{thm:reach-1clock}) that a matching \NP upper bound even holds for the \emph{reachability} problem for $1$-\TBPP.
Let us first motivate the key idea behind the construction. 
Consider the following \TBPP coverability query:
\begin{align}
    \label{eq:cov:query}
    \tag{\dag}
    (S, 0) \rcovers (A, 0)\msplus(B, 0).
\end{align}
If \eqref{eq:cov:query} holds,
then 
there is a derivation tree witnessing that $(S, 0) \reach{*} (A, 0) \msplus (B, 0) \msplus \gamma$
for some configuration $\gamma$.
The least common ancestor of leaves $(A,0)$ and $(B,0)$
is some process $(C,c)\in (\?X\x\nnreals)$.
Consider the \TA $\Aut$ obtained from the \TBPP by replacing branching rules $X \Rule{\varphi;R} X_i \msplus X_j$ with linear rules $X \Rule{\varphi;R} X_i$ 
and $X \Rule{\varphi;R} X_j$,
and let the reachability relation of $\Aut$ be expressed by \ELA formulas $\set{\varphi_{XY}}_{X, Y \in \?X}$,
which are of polynomial size by \cref{thm:1-clock:TA}.
Then our original coverability query \eqref{eq:cov:query} is equivalent to satisfiability the following \ELA formula:
\begin{align*}
    \psi \;\equiv\;
    \exists t_0, t_1, c \in \R \cdot \left(
        \varphi_{SC}(0, t_0, c) ~\land
            \varphi_{CA}(c, t_1, 0) \land
                \varphi_{CB}(c, t_1, 0) 
            \right).
\end{align*}
More generally, for any coverability query $(S, 0) \rcovers \alpha$
the number of common ancestors is linear in $\size \alpha$,
and thus we obtain a \ELA formula $\psi$ of polynomial size,
whose satisfiability we can check in \NP thanks to \cref{thm:logic:NP}.

\begin{theorem}
    \label{thm:cover-1clock}
    The coverability problem for 1-clock \TBPP is \NP-complete.
\end{theorem}

In order to witness reachability instances we need to refine the argument above
to restrict the \TA in such a way that they do not accidentally produce processes
that cannot be removed in time. To illustrate this point, consider a
1-\TBPP with rules
$$
    X\Rule{x=0}Y\msplus Z
    \qquad \textrm{ and } \qquad
    Z\Rule{x>0}\zerro.
$$
Clearly
$(X,0)\reach{*}(Y,0)$ holds in the \TA with rules
$X\Rule{x=0}Y$ and $X\Rule{x=0}Z$ instead of the branching rule above.
In the \TBPP however,
$(X,0)$ cannot reach $(Y,0)$ because the
branching rule produces a process $(Z, 0)$,
which needs a positive amount of time to be rewritten to $\zerro$.

\begin{definition}
    For a nonterminal $X$ let
    $\VAN{X}\subseteq\R^2$ be the binary predicate such that
    $$\VAN{X}(x,t) \quad \textrm{ if } \quad (X,x)\reach {t}\zerro$$
\end{definition}
\noindent
Intuitively, $\VAN{X}(x,t)$ holds if the configuration $(X,x)$ can vanish in time at most $t$.

The time 
it takes to remove a processes $(Z,z)$
can be computed as the value of a one-clock priced timed game \cite{BLMR2006, HIM2013}.
These are two-player games played on $1$-clock \TA
where players aim to minimize/maximize the cost of a play leading up to a designated target state.
Nonnegative costs may be incurred either by taking transitions, or by letting time elapse.
In the latter case, the incurred cost is a linear function of time, determined by the current control-state.
Bouyer et al.~\cite{BLMR2006} prove that such games admit $\eps$-optimal strategies for both players, so have well-defined cost value functions
determining the best cost as a function of control-states and clock valuation.
They prove that these value functions are in fact piecewise-linear.
Hansen et al.~\cite{HIM2013} later show that the piecewise-linear description
has only polynomially many line segments and can be computed in polynomial time
\footnote{This observation was already made, without proof, in \cite[Sec.~7.2.2]{T2009}.}.
We derive the following lemma. 
\begin{lemma}
    \label{lem:rasmus}
    A \QFLA formula expressing $\VAN{X}$ is effectively computable in polynomial time.
    More precisely, there is a
    set $\?I$ of polynomially many consecutive intervals
    $\{a_0\}(a_0,a_1)\{a_1\}(a_1,a_2)\{a_2\},\ldots (a_k,\infty)$
so that 
$$\VAN{X}(x,t) \ \equiv\ \bigvee_{0\le i \le k} (x=a_i\land t\ge_i c_i)\ \lor\ (a_i < x < a_{i+1} \land t \ge_i c_i - b_i x),$$
where 
the $a_i,c_i\in\R$ can be represented using polynomially many bits,
$\ge_i \in\{\ge,>\}$ and 
$b_i\in\{0,1\}$,
for all $0\le i<k$.
\end{lemma}
\begin{proof}[Proof (Sketch)]
    One can construct a one-clock priced timed game
    in which minimizer's strategies correspond to derivation trees.
    To do this, let unary rules $X\Rule{\varphi; R}Y$ carry over as transitions between (minimizer) states $X,Y$;
    vanishing rules $X\Rule{\varphi}\emptyset$ are replaced by transitions leading to a new target state $\bot$, which has a clock-resetting self-loop. 
    Branching rules $X\Rule{\varphi; R}Y\msplus Z$ can be implemented by 
    rules $X\Rule{\varphi}[Y,Z,\varphi]$, $[Y,Z,\varphi]\Rule{\varphi; R}Y$ and $[Y,Z,\varphi]\Rule{\varphi; R}Z$,
    where $X,Y,Z$ are minimizer states and $[Y,Z,\varphi]$ is a maximizer state.
    The cost of staying in a state is $1$,
    transitions carry no costs.
    Moreover, we need to prevent maximizer from elapsing time from the states she controls.
    For this reason, we consider an extension of price timed games where maximizer cannot elapse time.
    In the constructed game, minimizer has a strategy to reach $(\bot,0)$ from $(X,x)$ at cost $t$
    iff $(X,x)\reach{t}\emptyset$. 
    The result now follows from~\cite[Theorem~4.11]{HIM2013}
    (with minor adaptations in order to consider the more restrictive case where maximizer cannot elapse time)
    that computes value functions for the cost of reachability in priced timed games.
    These are piecewise-linear with only polynomially many line segments
    of slopes $0$ or $1$ which allows to present $\VAN{X}$ in \QFLA as stated.
    
    A more detailed and direct construction can be found in \cref{sec:a_proof_of_lem}.
\end{proof}

\Cref{lem:rasmus} allows us to compute a polynomial number of intervals $\?I$ sufficient to describe the $\VAN{X}$ predicates.
We will call a pair $(X,I)\in \?X \x \?I$ a \emph{region} here. 
A crucial ingredient for our construction will be timed automata
that are restricted in which regions they are allowed to produce as side-effects.
To simplify notations let us assume w.l.o.g.~that the given \TBPP has no resets along branching rules.
Let $I(r) \in \?I$ denote the unique interval containing $r\in \nnreals$,
and for a subset $S\subseteq \?X\x\?I$ of regions
write ``$Z \in S$'' for the clock constraint
expressing that $(Z, I(x)) \in S$. More precisely,
$$ Z \in S \ \equiv\ \bigvee_{(Z,(a_i,a_{i+1}))\in S} a_i<x<a_{i+1}
\ \lor\ 
\bigvee_{(Z,\set{a_i}))\in S} a_i=x.$$

\newcommand{\RTA}[1]{\mathit{TA}_{#1}}

\begin{definition}
    Let $S\subseteq {\?X\x\?I}$ be a set of regions
    for the $1$-\TBPP $(\{x\},\?X,\?R)$.
    We define a timed automaton $\RTA{S} = (\{x\},\?X,\?R_S)$
    so that $\?R_S$ contains all of the rules in $\?R$ with rhs of size 1 and none of the vanishing rules.
    Moreover, every branching rule $X\Rule{\varphi}Y \msplus Z$ in $\?R$ introduces
    \begin{itemize}
        \item a rule $X\Rule{\varphi'} Y$ guarded by $\varphi' \eqdef \varphi\land Z \in S$, and 
        \item a rule $X\Rule{\varphi'} Z$ guarded by $\varphi' \eqdef \varphi\land Y \in S$.
    \end{itemize}
\end{definition}

\begin{theorem}\label{thm:reach-1clock}
    The reachability problem for 1-clock \TBPP is in \NP.
\end{theorem}
\begin{proof}
    Suppose that there is a derivation tree witnessing a positive instance of the reachability problem
    and
    so that all branches leading to targets have duration $\tau$.
    We can represent a node by a triple $(A,a,\hat{a})\in(\?X\x\R\x\R)$,
    where $(A, a)$ is a \TBPP process
    and the third component $\hat{a}$ is the total time elapsed so far.
    Call a node $(A,a,\hat{a})$ \emph{productive} if it lies on a branch from root to some target node.
    Naturally, every node $(A,a,\hat{a})$ has a unique region $(A,I(a))$
    associated with it.
    For a productive node let us write $$S(A,a,\hat{a})$$ for the set of regions of
    nodes which are descendants of $(A,a,\hat{a})$
    which are non-productive but have a productive parent.
    See \cref{fig:oneclock-reach} (left) for an illustration.
    Observe that

    \begin{enumerate}
        \item The sets $S(A,a,\hat{a})$ can only decrease along a branch from root to a target.
        \item 
            If $(A,a,\hat{a})$ has a productive descendant
            $(C,c,\hat{c})$ such that
            $S(A,a,\hat{a}) = S(C,c,\hat{c})$,
            then $(A,a)\reach{\hat{c}-\hat{a}}(C,c)$ in the timed automaton $\RTA{S}$.
        \item Suppose $(A,a,\hat{a})$ has only one productive child $(C,c,\hat{c})$
            and that $S(A,a,\hat{a}) \supset S(C,c,\hat{c})$.
            Then it must also have another child $(B,b,\hat{b})$ s.t.~$\VAN{B}(b,\tau-\hat{b})$ holds.
    \end{enumerate}
    
    \noindent
    The first two conditions are immediate from the definitions of $S$ and $\RTA{S}$.
    To see the third, note that the 
    $S(A,a,\hat{a}) \supset S(C,c,\hat{c})$ implies that $(A,a,\hat{a})$ has some non-productive descendant
    $(B,b,\hat{b})$ whose region $(B,I(b))$ is not in $S(C,c,\hat{c})$.
    Since $(C,c,\hat{c})$ is the only productive child, that descendant must already be a child of $(A,a,\hat{a})$.
    Finally, observe that every non-productive node $(B,b,\hat{b})$ satisfies $\VAN{B}(b,\tau-\hat{b})$,
    as otherwise one of its descendants is present at time $\tau$, and thus must be a target node,
    contradicting the non-productivity assumption.

    The conditions above allow us to use labelled trees of polynomial size as reachability witnesses:
    These witnesses are labelled trees as above where only as polynomial number of checkpoints along branches from root
    to target are kept:
    A checkpoint is either the least common ancestor of two target nodes
    (in which case a corresponding branching rule must exist),
    or otherwise it is a triple of nodes as described by condition (3), where a region $(B,I(b))$ is produced for the last time. The remaining paths between 
checkpoints are positive reachability instances of timed automata $\RTA{S}$,
    as in condition (2),
    where the bottom-most automata $\RTA{S}$ satisfy that if $(U,I)\in S$ then $(z,0) \in\VAN{U}$ for all $z\in I$.
    Cf.~\cref{fig:oneclock-reach} (right).
    Notice that the existence of a witness of this form is expressible as a polynomially large \ELA formula
    thanks to \cref{lem:rasmus,thm:1-clock:TA}.

    Clearly, every full derivation tree gives rise to a witness of this form.
    Conversely, assume a witness tree as above exists. One can build a partial derivation tree
    by unfolding all intermediate \TA paths between consecutive checkpoints.
    It remains to show that whenever 
    some $\RTA{S}$ uses a rule $A\Rule{\varphi}C$ originating from a \TBPP rule 
    $A\Rule{\varphi}B\msplus C$ to produce a productive node $(C,c,\hat{c})$ then the node $(B,b,\hat{b})$
    produced as side-effect can vanish in time $\tau-\hat{b}$, i.e., we have to show that then $\VAN{B}(b,\tau-\hat{b})$.
    
    W.l.o.g.~let $\VAN{B}(x,t) \equiv t\ge d-xf$
    (the case with $>$ is analogous)
    for some $d\in\nnreals$ and $f\in\{0,1\}$.
    Observe that the region $(B,I(b))$ is in $S$ by definition of $\RTA{S}$ 
    and that the witness contains a later node $(B,b',\hat{b}')$ with $\VAN{B}(b',\tau-\hat{b}')$, and thus
    $$\tau-\hat{b}'\ge d-b'f.$$
    Notice also that $b'\ge b + \hat{b}'-\hat{b}$ as in the worst-case no reset appears on the path between
    the parent of $(B,b,\hat{b})$ and $(B,b',\hat{b}')$.
    Together with the inequality above we derive that
    $\tau-\hat{b}\ge d-bf$, meaning that indeed $\VAN{B}(b,\tau-\hat{b})$ holds, as required.
\end{proof}

\begin{figure}
\begin{subfigure}[B]{0.5\textwidth}
\centering
\includegraphics[height=4.4cm]{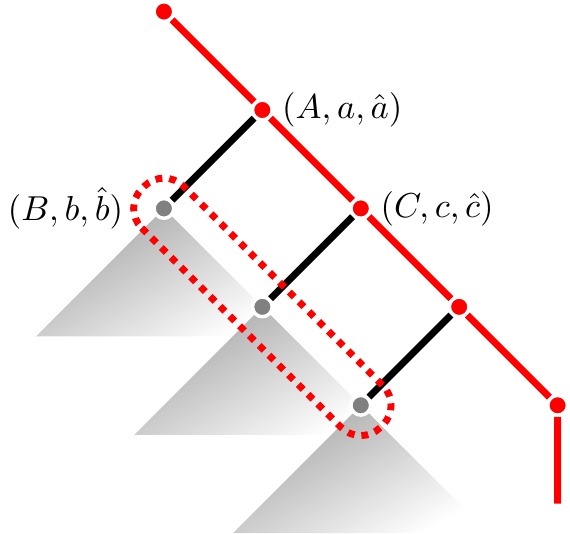}
\end{subfigure}\hfill
\begin{subfigure}[B]{0.5\textwidth}
\centering
\includegraphics[height=5.5cm]{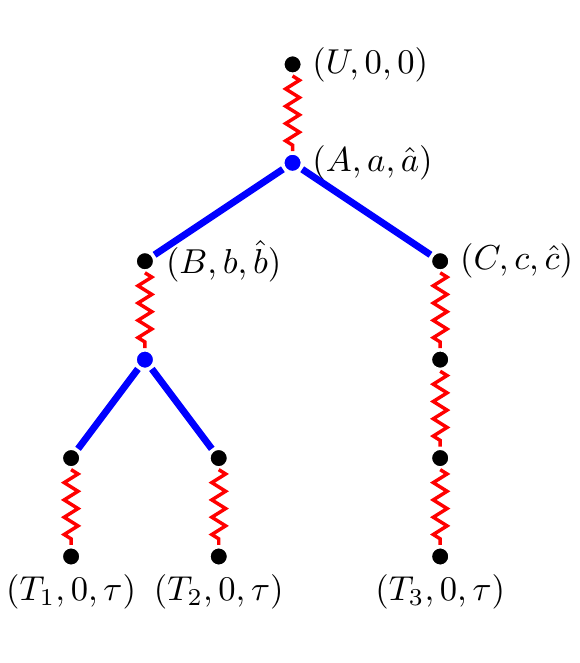}
\end{subfigure}
\caption{\textbf{Left:} Nodes on the red branch are productive, grey sub-trees are non-productive.
$S(A,a,\hat{a})$ contains the regions of nodes in the dotted region.
It holds that
$S(A,a,\hat{a}) \supseteq S(C,c,\hat{c})$
and the inequality is strict iff $(B,I(b))\in S(C,c,\hat{c})$.
\textbf{Right:} small reachability witnesses contain checkpoint where two productive branches split (in blue)
or where the allowed side-effects $S$ strictly decrease (red). The intermediate paths are runs of $S$-restricted \TA.
}
\label{fig:oneclock-reach}
\end{figure}

\section{Multi-Clock \TBPP}
\label{sec:multiclock}
In this section we consider the complexities of coverability and reachability problems for
\TBPP with multiple clocks.
For the upper bounds we will reduce to the reachability problem for \TA \cite{AD1994} and to solving reachability games for \TA \cite{JT2007}.

\begin{theorem}
    \label{thm:multiclock:cover}
    The coverability problem for $k$-\TBPP with $k \geq 2$ clocks is \PSPACE-complete.
\end{theorem}
\begin{proof}
    The lower bound already holds for the reachability problem of $2$-clock \TA \cite{FJ2015}
    and hence for the \emph{simple} \TBPP coverability.
    For the upper bound,
    consider an instance where $\Aut= (\?C,\?X,\?R)$ is a $k$-\TBPP and $T_1,\ldots,T_m$ are the target nonterminals.
    We reduce to the reachability problem for \TA $\Baut = (\?C',\?X',\?R')$ with exponentially many control states $\?X'$,
    but only $\size {\?C'} = O(k \cdot \size{\?X} \cdot m)$ many clocks.
    The result then follows by the classical region construction of \cite{AD1994},
    which requires space logarithmic in the number of nonterminals and polynomial in the number of clocks.
    The main idea of this construction is
    to introduce (exponentially many) new nonterminals and rules to simulate
    the original behaviour on bounded configurations only.

    Let $n = m + 2$.
    We have a clock $x_{X, i} \in \?C'$ for every original clock $x \in \?C$, nonterminal $X \in \?X$, and index $1 \leq i \leq n$,
    and a nonterminal of the form $[\alpha] \in \?R'$ for every multiset $\alpha \in\msets{\?X}$ of size at most $\size \alpha \leq n$.
    Since we are solving the coverability problem,
    we do not need to address vanishing rules $X \Rule {\varphi; R} \zerro$ in $\?R$, which are ignored.
    We will use clock assignments $S_{X, i} \equiv \bigwedge_{i \leq j \leq n - 1} x_{X, j} := x_{X, j + 1}$
    shifting by one position the clocks corresponding to occurrences $j = i, i+1, \dots, n - 1$ of $X$.
    We have three families of rules:
    \begin{enumerate}

        \item \textbf{(unary rules).} For each rule $X \Rule {\varphi; R} Y$ in $\?R$
        and multiset $\beta \in \msets{\?X}$ of the form $\beta = \gamma \msplus X \msplus \delta$
        of size $\size \beta \leq n$,
        for some $\gamma, \delta \in \msets{\?X}$,
        we have a corresponding rule  in $\?R'$
        $$[\beta] \Rule {\restrict \varphi {X, i}; \restrict R {X, i; Y, j}; S_{X, i}} [\gamma\msplus Y\msplus \delta]$$
        for every occurrence $1 \leq i \leq \beta(X)$ of $X$ in $\beta$ and for $j = \beta(Y) + 1$,
        where $\restrict \varphi {X, i}$ is obtained from $\varphi$ by replacing each clock $x$ with $x_{X, i}$,
        and $\restrict R {X, i; Y, j}$ is obtained from $R$ by replacing every assignment $x := y$
        by $x_{Y, j} := x_{X, i}$, and $x := 0$ by $x_{Y, j} := 0$.
    
    \item \textbf{(branching rules).} Let $X\Rule {} Y \msplus Z$ in $\?R$ be a branching rule.
        We assume w.l.o.g.~that it has no tests and no assignments, and that $X, Y, Z$ are pairwise distinct.
        We add rules in $\?R'$
        $$[\alpha \msplus X]\Rule {R; S_{X, i}} [\beta], \quad \textrm {with } \beta = \alpha \msplus Y \msplus Z \textrm{ and } \size \beta \leq n, $$
        for all $1 \leq i \leq \alpha(X)$ and $\alpha \in \msets{\?X}$,
        where $R \equiv \bigwedge_{x \in \?C} x_{Y, \beta(Y)} := x_{X, i} \wedge x_{Z, \beta(Z)} := x_{X, i}$
        copies each clock $x_{X, i}$ into $x_{Y, \beta(Y)}$ and $x_{Z, \beta(Z)}$,
        and $S_{X, i}$ was defined earlier.

    \item \textbf{(shrinking rules).} We also add rules that remove unnecessary nonterminals:
        For every $\beta = \alpha \msplus X \in \msets{\?X}$ with $\size \beta \leq n$
        and index $1 \leq i \leq \beta(X)$ denoting which occurrence of $X$ in $\beta$ we want to remove,
        we have a rule $[\alpha \msplus X] \Rule {S_i} [\alpha]$ in $\?R'$.
        
    \end{enumerate}

    It remains to argue that
    $(X,\vec{0})\reach{*}(T_1,\vec{0})\msplus \cdots \msplus(T_n,\vec{0})$ in $\Aut$
    if, and only if,
    $([X],\vec{0})\reach{*}([T_1 \msplus \cdots \msplus T_n], \vec{0})$ in $\Baut$.
    This can be proven via induction on the depth of the derivation tree, where
    the induction hypothesis is that every configuration $\alpha$ of size at most $n$
    can be covered in with a derivation tree of depth $d$ in $\Aut$ if, and only if,
    in the timed automaton $\Baut$ the configuration $[\alpha]$ can be reached via a path of length at most $d$.
\end{proof}

\ignore{
\subsection{Coverability for 1-clock \TCFP}

\ilorenzo{Move to a later section}

\begin{theorem}
    The general coverability problem for 1-clock \TCFP is \NP-complete.
    \NP-hardness holds already when the target set has size $\geq 2$.
\end{theorem}

\NP-hardness follows from the proof of Theorem~\ref{thm:reachability:1-clock},
since the reduction from subset sum presented therein 
produces a 1-clock \TCFP for which reachability and coverability are equivalent, i.e.,
$(S, 0) \goesto * (X_{k+1}, 0) (Z, 0)$ iff $(S, 0) \covers (X_{k+1}, 0) (Z, 0)$.

For the \NP upper bound we use the following simple observation.
Consider the following \TCFP coverability query:
\begin{align}
    \label{eq:cov:query}
    \tag{\dag}
    (S, 0) \covers (A, 0) (B, 0).
\end{align}
If \eqref{eq:cov:query} holds in the \TCFP,
then in the underlying \emph{untimed \CFG}
there is a derivation $S \goesto * x A y B z$, for some $x, y, z \in \mathcal X^*$.
Let $C \in \mathcal X$ be the least common ancestor of $A$ and $B$ in the corresponding derivation tree,
Consider the \TA obtained from the \TCFP by replacing branching rules $X \goesto {} YZ$ with linear rules $X \goesto {} Y$ and $X \goesto {} Z$.
and let $\varphi_{XY}(x, t, x')$ be the \EHA formula describing the reachability relation in the \TA,
which is of polynomial size by Theorem~\ref{thm:1-clock:TA}.
Then our original coverability query \eqref{eq:cov:query} is equivalent to the following \EHA formula:
\begin{align*}
    \psi \;\equiv\;
        \exists t_0, t_1, t_2 \in \R \cdot
        \varphi_{SC}(0, t_0, x_C) \wedge
            \varphi_{CA}(x_C, t_1, 0) \wedge
                \varphi_{CB}(x_C, t_2, 0) \wedge
                    t_0 + t_1 = t_0 + t_2.
\end{align*}
For a general coverability query $(S, 0) \covers (\alpha, \bar 0)$
the number of common ancestors $C$'s is linear in $\size \alpha$,
and thus we obtain a \EHA formula $\psi$ of polynomial size.
Thanks to Theorem~\ref{thm:logic:NP} we can check whether it is satisfiable in \NP,
as required.

\subsubsection{Simple coverabiltiy for 1-clock \TCFP.}

\ilorenzo{Move to a later section}

If we restrict our attention to target configurations of size 1,
then the complexity of the coverability problem improves to \NL.
This holds simply because, in this case,
we can replace a branching rule of the form $X \goesto Y Z$
with two linear rules $X \goesto Y$ and $X \goesto Z$,
thus yielding a 1-clock \TA, for which reachabiltiy is in \NL.

\begin{theorem}
    The simple coverability problem for 1-clock \TCFP is \NL-complete.
\end{theorem}

}

\begin{theorem}
\label{thm:multiclock:reach}
    The reachability problem for \TBPP is \EXPTIME-complete.
    Moreover, \EXPTIME-hardness already holds 
    for $k$-\TBPP emptiness, if 
    1) $k\ge 2$ is any fixed number of clocks, or
    2) $k$ is part of the input but only $0$ or $1$ appear as constants in clock constraints.
\end{theorem}
In the remainder of this section contains a proof of this result, in three steps:
In the first step 
(\cref{lem:multiclock:reach-1})
we show an \EXPTIME upper bound for the special case of \emph{simple reachability},
i.e., when the target configuration has size $1$.
As a second step (\cref{lem:multiclock:reach-2})
we reduce general case to simple reachability and thereby prove the upper bound claimed in \cref{thm:multiclock:reach}.
As a third step (\cref{lem:multiclock:reach-3}), we prove the corresponding lower bound.

\begin{restatable}[Simple reachability]{lemma}{simpleReachability}
    \label{thm:simple:reachability}
    \label{lem:multiclock:reach-1}
    The simple reachability problem for \TBPP is in \EXPTIME.
    More precisely, the complexity is exponential in the number of clocks and the maximal clock constant,
    and polynomial in the number of nonterminals.
\end{restatable}

\begin{proof}[Proof (Sketch)]
    We reduce to \TA reachability games,
    where two players (Min and Max) alternatingly determine a path of a \TA,
    by letting the player who owns the current nonterminal pick time elapse and a valid successor configuration.
    Min and Max aim to minimize/maximize the time until the play first visits a target nonterminal $T$.
    \TA reachability games can be solved in \EXPTIME, with the precise time complexity claimed above
    \cite[Theorem 5]{JT2007}.
    The idea of the construction is to let Min produce a derivation tree along the branch that leads to (unique) target process.
    Whenever she proposes to go along branching rule,
    Max gets to claim that the other sibling, not on the main branch, cannot be removed until the main branch ends.
    This can be faithfully implemented by storing only the current configuration on the main branch
    plus one more configuration (of Max's choosing) that takes the longest time to vanish.
    Min can develop both independently but must apply time delays to both simultaneously.
    Min wins the game if she can reach the target nonterminal and before that moment all the other branches have vanished.
\end{proof}

\begin{lemma}
\label{lem:multiclock:reach-2}
    The reachability problem for \TBPP is in \EXPTIME.
\end{lemma}

\begin{proof}
First notice that the special case of reachability of the empty target set
trivially reduces to the simple reachability problem
by adding a dummy nonterminal, which is created once at the beginning and has to be the only
one left at the end.
Suppose we have an instance of the $k$-\TBPP reachability problem 
with target nonterminals $T_1,T_2,\ldots,T_m$.
We will create an instance of simple reachability
where the number of nonterminals increases exponentially but the number of clocks is $O(k \cdot \size{\?X} \cdot m)$.
In both cases, the claim follows from \cref{thm:simple:reachability}.

We introduce a nonterminal $[\beta]$ for every multiset $\beta \in \msets{\?X}$ of size $\size\beta \leq n := m+2$,
and we have the same three family of rules as in proof of \cref{thm:multiclock:cover},
where the last family 3.~is replaced by the family below:
\begin{description}
    \item[3'.] We add extra branching rules in order to maintain nonterminals $[\beta]$ corresponding to small multisets $\size\beta \leq n$.
    Let $\beta \in \msets{\?X}$ of size $\size \beta \leq n$
    and consider a partitioning $\beta = \beta_1 \msplus \beta_2$,
    for some $\beta_1, \beta_2 \in \msets{\?X}$.
    We identify $\beta$ with the set $\beta = \setof{(X, i)} {X \in \?X, 1 \leq i \leq \beta(X)}$
    of pairs $(X, i)$, where $i$ denotes the $i$-th occurrence of $X$ in $\beta$ (if any),
    and similarly for $\beta_1, \beta_2$.
    We add a branching rule
    $$[\beta] \Rule {} (\beta, f, \beta_1) \msplus (\beta, f, \beta_2), $$
    where $(\beta, f, \beta_i)$ are intermediate locations,
    for every bijection $f : \beta \to \beta_1 \cup \beta_2$
    assigning an occurrence of $X$ in $\beta$ to an occurrence of $X$ either in $\beta_1$ or $\beta_2$.
    We then have clock reassigning (non-branching) rules
    $$(\beta, f, \beta_1) \Rule {S_1} [\beta_1] \quad \textrm{and} \quad
    (\beta, f,  \beta_2) \Rule {S_2} [\beta_2],$$
    where $S_1 \equiv \bigwedge_{x \in \?C} \bigwedge_{X \in \?X} \bigwedge_{1 \leq i \leq \beta_1(X)} x_{X, i} := x_{f^{-1}(X, i)}$
    and similarly for $S_2$. \qedhere
\end{description}
\end{proof}

\begin{restatable}{lemma}{multiclockEXPTIMEhardness}
    \label{lem:multiclock:reach-3}
    The non-emptiness problem for \TBPP
    is \EXPTIME-hard already in discrete time, for
    1) \TBPP with constants in $\set{0, 1}$ (where the number of clocks is part of the input),
    and 
    2) for $k$-\TBPP for every fixed number of clocks $k \geq 2$.
\end{restatable}

\section{Conclusion}
\label{sec:conclusion}
We introduced basic parallel processes extended with global time
and studied the complexities of several natural decision problems, including variants of the coverability and reachability problems.
\Cref{tab:tbpp} summarizes our findings.

The exact complexity status of the simple reachability problem for $1$-\TBPP is left open.
An \NP upper bound holds from the (general) reachability problem (by \cref{thm:reach-1clock})
and \P-hardness comes from the emptiness problem for context-free grammars.
We conjecture that a matching polynomial-time upper bound holds.

Also left open for future work are \emph{succinct} versions the coverability and reachability problems, 
where the target size is given in binary.
A reduction from subset-sum games \cite{FJ2015}
shows that the succinct coverability problem for 1-\TBPP is \PSPACE-hard.
This implies that our technique showing the \NP-membership for the non-succinct version of the coverability problem (cf.~\cref{thm:cover-1clock})
does not extend to the succinct variant, and new ideas are needed.

\newcommand{\red}{\cellcolor{red}}
\newcommand{\orange}{\cellcolor{orange}}
\newcommand{\REF}[1]{(#1)}
\begin{table}[htpb]
    \centering
    \begin{tabular}{r|c|c|c|c|c}
                &Emptiness
                &\makecell{Simple \\ Coverability}
                &Coverability
                &\makecell{Simple \\Reachability}
                &Reachability
                \\
        \hline
        \TBPP   
                &\makecell{\EXPTIME\\~[Lem~\ref{lem:multiclock:reach-3}], \cite{JT2007}}
                &\makecell{\PSPACE \\~[Thm~\ref{thm:multiclock:cover}]}
                &\makecell{\PSPACE \\~[Thm~\ref{thm:multiclock:cover}]}
                &\makecell{\EXPTIME\\~[Thm~\ref{thm:simple:reachability}]}
                &\makecell{\EXPTIME\\~[Thm~\ref{thm:multiclock:reach}]}
        \\
        \hline
        1-\TBPP 
                &\makecell{\P \\~\cite{LMS2004} \ignore{\cite{HIM2013}}}
                &\makecell{\NL\\~\cite{LMS2004}}
                &\makecell{\NP \\~[Thm~\ref{thm:cover-1clock}]}
                &\P~/ \NP
                &\makecell{\NP \\~[Thm~\ref{thm:reach-1clock}]}
                \\
    \end{tabular}
        \\[2ex]
    \caption{Results on \TBPP and $1$-clock \TBPP.
    The decision problems are complete for the stated complexity class. Simple Coverability/Reachability
    refer to the variants where the target has size 1.}
    \label{tab:tbpp}
\end{table}

\bibliography{journals,conferences,autocleaned}
\appendix

\section{Missing proofs for \cref{sec:ha}}

\label{sec:logic:app}

\logicNP*

\noindent
It was proved in \cite{W1999} that \LA admits quantifier elimination
by reduction to quantifier elimination procedures for \PA and \RA.
However, the complexity of deciding the existential fragment of \LA was not discussed,
and moreover the procedure translating a formula of \LA to the separated fragment
has an exponential blow-up when dealing with modulo constraints in the proof of~\cite[Theorem~3.1]{W1999} (case 2),
and moreover it relies on an intermediate transformation with no complexity bound~\cite[Lemma~3.3]{W1999}.
Our result was obtained independently from \cite{BJW:ACM:2005},
where a similar argument is given showing that each sentence of \LA
can be transformed in polynomial time into a logically equivalent boolean combination of \PA and \RA sentences.
In fact, investigating the proof of \cite[Theorem 3.1]{BJW:ACM:2005}
one can see that their translation even preserves existential formulas.
We give below a short self-contained argument reducing \LA to \PA and \RA with only a polynomial blowup,
with the further property that it preserves existential formulas.

\begin{proof}
	By introducing linearly many new existentially quantified variables and suitable defining equalities,
	we assume w.l.o.g.~that terms are \emph{shallow}:
	\begin{align*}
		s, t \; ::= \; x \sep k \sep \floor x \sep \fractof x \sep {-x} \sep x + y \sep k \cdot x.
	\end{align*}
	Since now we have atomic propositions of the form $s \leq t$ 
	with $s, t$ shallow terms,
	we assume that we have no terms of the form ``$-x$''
	(by moving it to the other side of the relation
	and possibility introducing a new existential variable to make the term shallow again).
	Moreover, we can also eliminate terms of the form $k\cdot x$
	by introducing $O(\log k)$ new existential variables and using iterated doubling (based on the binary expansion of $k$);
	for instance, $5\cdot x$ is replaced by $x_0 + x_2$,
	by adding new variables $x_0, \dots, x_2$ and equalities
	$x_0 = x$, $x_1 = x_0 + x_0$, $x_2 = x_1 + x_1$.
	We end up with the following further restricted syntax of terms:
	\begin{align*}
		s, t \; ::= \; x \sep k \sep \floor x \sep \fractof x \sep x + y.
	\end{align*}
	We can now expand atomic propositions
	by replacing the expression on the left with the equivalent one on the right:
	\begin{align*}
		s \leq t \qquad \textrm{ iff } \qquad
			\floor s < \floor t \vee
			(\floor s = \floor t \wedge \fractof s \leq \fractof t),
	\end{align*}
	We push the integral $\floor \_$ and fractional $\fractof \_$ operations inside terms,
	according to the following rules:
	\begin{align*}
		\floor k 
			&\to k &
		\fractof k
			&\to 0	
		\\
		\floor {\floor x}
			&\to \floor x &
		\fractof {\floor x}
			&\to 0	
		\\
		\floor {\fractof x}
			&\to 0 &
		\fractof {\fractof x} 
			&\to \fractof x
	\end{align*}
	It remains to consider sums $x + y$.
	We perform a case analysis on $\fractof x + \fractof y$:
	\begin{align*}
		s \sim \floor {x + y}
			\to\ 	&(s \sim \floor x + \floor y \wedge \fractof x + \fractof y < 1)\ \vee \\
					&(s \sim \floor x + \floor y + 1 \wedge \fractof x + \fractof y \geq 1), \\
		s \sim \fractof {x + y}
			\to\ 	&(s \sim \fractof x + \fractof y \wedge \fractof x + \fractof y < 1)\ \vee \\
					&(s \sim \fractof x + \fractof y - 1 \wedge \fractof x + \fractof y \geq 1).
	\end{align*}
	We thus obtain a logically equivalent \emph{separated formula},
	i.e., one where integral $\floor x$ and fractional $\fractof y$ variables never appear together in the same term.
	Since we only added existentially quantified variables in the process,
	the resulting formula is still in the existential fragment,
	which can be decided in \NP by calling separately decision procedures for \PA and \RA.
\end{proof}

\TAreachLemma*

\begin{proof}
	Consider the following binary relation $R$
	between the configurations of the infinite transition system induced by the \TA $\Aut$
	and states of the \NFA $\Baut$:
	\begin{align}
		\label{eq:ta:bisim}
		((X, u), (X', \lambda)) \in R \quad \textrm{ iff } \quad
			X = X' \wedge u \in \lambda.
	\end{align}
	Let $\hat b = \trans$ for $b = (\trans, a)$ and $\hat \tau = \tau$ otherwise.
	It is immediate to show that $R$ is a variant of timed-abstract bisimulation, in the following sense:
	For every configuration $c$ of $\sem \Aut$ and state $d$ of $\Baut$,
	if $(c, d) \in R$, then
	\begin{enumerate}
		\item
			For every transition of $\Aut$ of the form $c \goesto \trans c'$
			there is a transition of $\Baut$ of the form $d \goesto b d'$
			s.t.~$\hat b = \trans$ and $(c', d') \in R$.
			Moreover, if $\trans = \tau$ is a time elapse,
			then in fact $d \goesto \tau d'$ \emph{for every} $d'$ s.t.~$(c', d') \in R$.
		\item
			For every transition of $\Baut$ of the form $d \goesto b d'$
			there is a transition of $\sem \Aut$ of the form $c \goesto \trans c'$
			s.t.~$\hat b = \trans$ and $(c', d') \in R$.
			Moreover, if $b = \tau$ is a symbolic time elapse,
			then in fact $c \goesto \tau c'$ \emph{for every} $c'$ s.t.~$(c', d') \in R$.
	\end{enumerate}
	The two additional conditions in each point distinguish $R$ from timed-abstract bisimulation.

	For the ``only if'' direction, assume $u \tareach {XY} \delta v$,
	as witnessed by a path
	$$(X_0, u_0) \step {\trans_1 \delta_1} (X_1, u_1) \step {\trans_2 \delta_2} \cdots \step {\trans_n \delta_n} (X_n, u_n), $$
	with 
	$(X_0, u_0) = (X, u)$ and $(X_n, u_n) = (Y, v)$,
	where $u, v$ are the initial, resp., final values of the clock,
	and $\delta = \delta_0 + \cdots + \delta_n$ is the total time elapsed.
	Let $\lambda, \mu \in \Lambda$ be the unique intervals s.t.~$u \in \lambda, v \in \mu$,
	and take $c = (X, \lambda), d = (Y, \mu)$.
	Since $((X, u), c) \in R$, by the definition of $R$ there exists a corresponding path in $\Baut$
	$$c_0 \goesto {b_1 \tau} c_1 \goesto  {b_2 \tau} \cdots \goesto {b_n \tau} c_n, $$
	where $c_0 = (X, \lambda)$, $c_n = (Y, \mu)$, for every $1 \leq i \leq n$, $((X_i, u_i), c_i) \in R$,
	$c_i$ of the form $(X_i, \lambda_i)$,
	and $b_i$ is of the form $(\trans_i, a_i)$,
	with $a_i = \tick {\lambda_i}$ whenever $\trans_i$ is a reset (and $a_i = \varepsilon$ otherwise).
	Let $\trans_{i_1}, \dots, \trans_{i_m}$ be all reset transitions in $\trans_1, \dots, \trans_n$.
	The time elapsed between the $j$-th and the $(j+1)$-th reset is
	$u_{i_{j+1} - 1} = \delta_{i_j} + \delta_{i_j + 1} + \cdots + \delta_{i_{j+1} - 1}$,
	and, by the invariant above, $u_{i_{j+1} - 1} \in \lambda_{i_{j+1} - 1}$
	and $a_{i_{j+1} - 1} = \tick {\lambda_{i_{j+1} - 1}}$.
	Consequently, the total time elapses is
	$\delta = (u_{i_1 - 1} - u_0) + (u_{i_2 - 1} + \cdots + u_{i_m - 1}) + u_m$.
	By the definition of $\varphi_{L_{cd}}$,
	this shows $(u, \delta, v) \models x \in \lambda \wedge x' \in \mu \wedge \psi_{cd}(x, \ttime, x')$, as required.
	
	For the ``if'' direction, assume $(u, \delta, v) \models \varphi_{XY}(x, \ttime, x')$.
	There exist configurations $c = (X, \lambda)$ and $d = (Y, \mu)$ s.t.~
	$u \in \lambda$, $v \in \mu$, and $(u, \delta, v) \models \psi_{cd}(x, \ttime, x')$.
	By identifying variables with their values to simplify the notation in the sequel,
	there exist $\vec y = (y_\lambda)_{\lambda \in \Lambda}$ s.t.\ $\vec y \models \psi_{L_{cd}}$
	and $\vec x = (x_\lambda)_{\lambda \in \Lambda}$
	s.t.\ $x_\lambda \in y_\lambda \cdot \lambda$ and $\delta = v - u + \sum_{\lambda \in \Lambda} x_\lambda$.
	By the definition of $\psi_{L_{cd}}$, there exists a run in $\Baut$ from $c$ to $d$ of the form
	$$ (X, \lambda) \goesto {\tau (\trans_1, a_1)} (X_1, \lambda_1) \goesto {\tau (\trans_2, a_2)} \cdots \goesto {\tau (\trans_n, a_n)} (Y, \mu),$$
	where without loss of generality
	we have composed possibly several $\tau$-transitions together into a single $\tau$-transition,
	\st, whenever $\trans_i$ is a reset, $a_i = \tick {\lambda_i}$ (and $a_i = \varepsilon$ otherwise).
	Consequently, $y_\lambda$ is precisely the number of $\trans_i$'s s.t.~$a_i = \tick \lambda$.
	Since $((X, u), c) \in R$ and $R$ is a timed-abstract bisimulation,
	there exists a corresponding run in $\sem \Aut$
	from $(X, u)$ to $(Y, v)$ of the form
	$$ (X, u) \goesto {\delta_1 \trans_1} (X_1, u_1) \goesto {\delta_2 \trans_2} \cdots \goesto {\delta_n \trans_n} (Y, v)$$
	s.t.~$((X_i, u_i), c_i) \in R$ for every $1 \leq i < n$, i.e., $u_i \in \lambda_i$.
	It remains to show that $\delta$ is the total time elapsed by the run above,
	i.e., $\delta = \delta_1 + \cdots + \delta_n$.
	Let $\trans_{i_1}, \dots, \trans_{i_m}$ be all reset transitions from $\trans_1, \dots, \trans_n$.
	By the definition of $R$, we can choose the $u_i$'s arbitrarily in $\lambda_i$.
	In particular we can choose $u_{i_1} \in \lambda_{i_1}, \dots, u_{i_m} \in \lambda_{i_m}$
	at the time of resets s.t.~for every interval $\lambda \in \Lambda$,
	the cumulative value of the clock at the times of reset when it was in interval $\lambda$
    is precisely $x_\lambda = \sum \setof {u_{i_j}} {\lambda_{i_j} = \lambda}$.
	Consequently, the total time accumulated by the clock in any reset is
	$\sum_{\lambda \in \Lambda} x_\lambda = u_{i_1} + \cdots + u_{i_m}$.
	Moreover, we have $u_{i_1} = \delta_1 + \cdots + \delta_{i_1} - u$,
	$u_{i_j} = \delta_{i_{j-1} + 1} + \cdots + \delta_{i_k}$ for $1 < j \leq m$,
	and $v = \delta_{i_m + 1} + \cdots + \delta_n$.
	Consequently, $\delta_1 + \cdots + \delta_n = u - v + \sum_{\lambda \in \Lambda} x_\lambda$,
	as required.
\end{proof}

\section{Missing proofs for \cref{sec:multiclock}}

\simpleReachability*

\begin{proof}
    Let $\Baut = (\?C, \?X, \?R)$ be the input \TBPP
    and suppose we want to solve a simple reachability query $(X, \vec 0) \reach * (T, \vec 0)$
    We denote with $\?X_\emptyset$ the set $\?X \cup \set \emptyset$.
    We construct the \TA reachability game $\Aut = (\?D, \?Y = \Ymin \cup \Ymax, \?S)$,
    where the set of clocks $\?D$ contains two copies $x_L$ and $x_R$ for every clock $x \in \?C$,
    locations in $\Ymin=\?X \times \?X_{\set\emptyset}$ are of the form $(X, Y_\emptyset)$ with $Y_\emptyset = Y \in \?X$ or $Y_\emptyset = \emptyset$,
    and $\Ymax=\?X\times \?X \times \?X_\emptyset$.
    For a formula $\varphi$ and $D \in \set{L, R}$, we write $\restrict \varphi D$
    for the formula obtained from $\varphi$ by replacing every clock $x\in\?C$ by its copy $x_D$;
    similarly for $\restrict S D$, where $S$ is a sequence of clock assignments.
    The initial location is $(X, \emptyset)$ and the target location is $(T, \emptyset)$.

    The idea of the construction is to let
    Min stepwise produce a derivation tree along the branch that leads to (unique) target process.
    Whenever she proposes to go along a branching rule,
    Max gets to claim that the other sibling, not on the main branch,
    cannot be removed until the branch ends.
    This can be faithfully implemented by storing only the current configuration on the main branch
    plus one more configuration (of Max's choosing) that takes the longest time to vanish.
    Min can develop both independently but must apply time delays to both simultaneously.
    The game ends and Min wins when $(T, \emptyset)$ is reached.

    A round of the game from position $(X, Y_\emptyset)$ is played as follows. 
    We assume w.l.o.g.~that there are only branching and vanishing rules.
    There are two cases.
    \begin{enumerate}

        \item In the first case, Min plays from the first component.
            For every rule $X \Rule {\varphi;S} Z_0\msplus Z_1$, 
            Min chooses that $Z_i$ is the nonterminal which will reach $T$, and that $Z_{1-i}$ should vanish.
            Thus, the game has two transitions
            $$X \Rule {\restrict \varphi L; \restrict S L} (Z_i, Z_{1-i}, Y_{\emptyset}), \quad \textrm{ for } i \in \set{0, 1},$$
            This is followed by Max, who chooses whether to keep the second component $Y_\emptyset$,
            or to replace it by $Z_{1-i}$.
            Thus, there are transitions
            $$(Z_i, Z_{1-i}, Y_{\emptyset}) \Rule{} (Z_i, Y_{\emptyset}) \quad \textrm{ and } \quad
            (Z_i, Z_{1-i}, Y_{\emptyset}) \Rule{} (Z_i, Z_{1-i}), \quad \textrm{ for } i \in \set{0, 1}.$$

        \item In the second case, Min plays from the second component $Y_\emptyset = Y \neq \emptyset$.
        There are two subcases.
        \begin{enumerate}
            \item Min selects a vanishing rule $Y \Rule {\varphi} \emptyset$.
            The corresponding transition in the game is
            $$(X, Y) \Rule {\restrict \varphi L} (X, \emptyset).$$
            \item Min selects a branching rule $Y \Rule {\varphi; S} Z_0\msplus Z_1$. 
            which corresponds to the following transition in the game:
            $$ (X, Y) \Rule {\restrict \varphi R; \restrict S R} (X, Z_0, Z_1). $$
            Then, Max chooses to keep $Z \in \set{Z_0, Z_1}$, corresponding to the following two transitions
            $$ (X, Z_0, Z_1) \Rule {} (X, Z), \quad \textrm{ for } Z \in \set{Z_0, Z_1}.$$
        \end{enumerate}
    \end{enumerate}
    We have that Min has a strategy to reach $(T, \emptyset)$ from $(X, \emptyset)$ in the \TA game if, and only if,
    $(X, \vec{0}) \reach{*} (T, \vec{0})$ in the \TBPP.
\end{proof}

\multiclockEXPTIMEhardness*

\mypar{First case.}
    The first case can be shown by reduction from the non-emptiness problem of the language intersection of one
    context-free grammar and several nondeterministic finite automata
since we can reduce to the case above by adding polynomially many clocks.
We reduce from the non-emptiness problem of the intersection of one \CFG $G$
and several nondeterministic finite automata (\NFA) $\Aut_1, \dots, \Aut_n$,
which is \EXPTIME-hard (this is a folklore result; cf.~also \cite{HLMS:LMCS:2012}).
We proceed as in a similar reduction for densely-timed pushdown automata \cite{AAS:LICS:2012}.
We assume \wlg that all $\Aut_i$'s control locations are from a fixed set of integers $Q = \set{0, \dots, n - 1}$.
We construct a \TBPP $\Baut$ with $2n + 1$ clocks:
For each \NFA $\Aut_i$, we have a clock $x_i$ representing its current control location
and a clock $y_i$ representing its future control location after the relevant nonterminal reduces to $\varepsilon$;
an extra clock $t$ guarantees that no time elapses.
We assume \wlg that $G$ is in Chomsky-Greibach normal form, i.e., productions are either of the form $X \goesto a YZ$ or $X \to \varepsilon$.

\ignore{
In the construction below we allow some syntactic sugar in order to simplify the presentation of the \TBPP.
We will use clock values only from the discrete and bounded domain $D = \set{0, 1, \dots, n - 1}$;
thus our reduction works already in discrete time.
We allow assignment operations of the form $x_i := k$ with $k \in D$.
Such an operation can be simulated by elapsing $n - k$ time units, resetting $x_i$, and then elapsing $k$ time units (thus $n$ time units in total).
Along the way, in order to preserve the value of the other clocks (also in $D$),
at every time unit we reset clocks that reach value $n$.
We will also use clock assignments of the form $x_i := y_j$.
They can be simulated by nondeterministically guessing the value $x_i = k$ and then assigning $y_j := k$.
}

We now describe the reduction.
The simulation of reading input symbol $a$ proceeds in $n+1$ steps.
In step $0$, for every production $X \goesto a YZ$ of $G$ we have a production of $\Baut$
\begin{align*}
	X \Rule {t = 0} Y_{a, 1}(\vec x := \vec x, \vec y := \vec k, t := t) \msplus Z(\vec x := \vec k, \vec y := \vec y, t := t).
\end{align*}
We use here an extended syntax to simplify the presentation.
Intuitively, the rule above requires guessing a new tuple of clock values $\vec k \in Q^n$,
which represent the intermediate control locations of the \NFA's after nonterminal $Y$ will vanish.
Process $Y_{a, 1}$ inherits the value of clocks $\vec x$ from $X$,
and the new value of clocks $\vec y$ is the guessed $\vec k$.
Symmetrically, process $Z$ inherits the value of clocks $\vec y$ from $X$,
and the new value of clocks $\vec x$  is $\vec k$.
While guessing $\vec k$ seemingly requires an exponential number of rules,
one can in fact implement it by guessing its components one after the other,
by introducing $n$ extra clocks; we avoid the extra bookkeeping for simplicity.

In step $i$, for $1 \leq i \leq n$, for every transition $p_i \goesto a q_i$ of $\Aut_i$, we have a production
\begin{align*}
    X_{a, i} \Rule {t = 0, x_i = p_i, x_i := q_i} X_{a, i+1}.
\end{align*}
The current phase ends with a production $X_{a, n+1} \Rule {t=0} X$.
Finally, a production $X \to \varepsilon$ of $G$ is simulated by
\begin{align*}
	X \Rule {\vec x = \vec y} \emptyset.
\end{align*}
The following lemma states the correctness of the construction.
We denote by $\lang G X$ the set of words over the terminal alphabet recognised by nonterminal $X$ in the grammar $G$,
and by $\lang {\Aut_i} {p, q}$, with $p, q \in Q_i$, the set of words s.t.~the \NFA $\Aut_i$ has a run from state $p$ to state $q$.

\begin{lemma}
    We have that $\lang G X \cap \lang {\Aut_1} {p_1, q_1} \cap \cdots \cap \lang {\Aut_n} {p_n, q_n} \neq \emptyset$
    if, and only if, $(X, \mu) \reach * \emptyset$,
    where $\mu = \set {x_1 \mapsto p_1, y_1 \mapsto q_1, \dots, x_n \mapsto p_n, y_n \mapsto q_n}$.
\end{lemma}

\mypar{Second case.}

We now show that the reachability of the empty configuration problem for $k$-\TBPP is \EXPTIME-hard 
for any fixed number of clocks $k \geq 2$.
This can be shown by reduction from \emph{countdown games} \cite{JLS2008},
which are two-player games $(Q,T,k)$ given by a finite set $Q$ of control states, a finite set $T \subseteq (Q \x \N_{>0} \x \Q)$ of transitions, labelled by positive integers, and a target number $k\in\N$.
All numbers are given in binary encoding.
The game is played in rounds, each of which starts in a pair $(p,n)$ where $p\in Q$ and $n \le k$,
as follows.
First player $0$ picks a number $l \le k-n$, so that at least one $(p,l,p')\in T$ exists;
Then player $1$ picks one such transition and the next round starts in $(p',n+l)$.
Player $0$ wins iff she can reach a configurations $(q,k)$ for some state $q$.

Determining the winner in a countdown game is \EXPTIME-complete \cite{JLS2008}
and can easily encoded as a reachability of the empty configuration problem for a $2$-\TBPP $(\{x_1,x_2\},Q, \?R)$.
All we need is one nonterminal per control state and two clocks. 
Suppose $p_1,p_2, \ldots ,p_m$ are the $l$-successors of state $p$ in the game. Then our \TBPP has a rule
$$p \Rule{x_1=l; x_1:=0} p_1\msplus p_2\msplus\ldots\msplus p_m$$
This checks that the first clock is $l$ and then resets it.
Finally, for all states $p$ there is a rule
$p \Rule{x_2=k} \zerro$
that checks if the second clock (which is never reset) equals the target $k$.
  
Notice that a derivation tree that reduces $(p,0,0)$ to the empty multiset is a winning strategy for player $0$ in the
countdown game: Since all branches end in a vanishing rule, the total time on the branch was exactly $k$.

\section{\PSPACE-hardness of $1$-\TBPP Coverability with Binary Targets}
\label{app:1-cover-binary}
We show that reachability and coverability for 1-clock \TBPP with target sets 
given as a vector in $\N^\?X$ in binary encoding, is \PSPACE-hard.
This contrasts with the case of targets encoded in unary, which is \NP-complete
(\cref{thm:multiclock:cover}).
\newcommand{\Adam}{$\forall$}
\newcommand{\Eve}{$\exists$}

We reduce from solving subset-sum games, which is \PSPACE-complete \cite{FJ2015}.
A \emph{subset-sum game} (\SSG) has as input a natural number $s \in \N$ and a tuple
\begin{align*}
    \forall \set {u_1, v_1} \exists \set {w_1, z_1} \cdots \forall \set {u_n, v_n} \exists \set {w_n, z_n},
\end{align*}
where all $u_i, v_i, w_i, z_i \in \N$ are given in binary.
The game is played by two players, \Adam\ and \Eve\,
who alternate in choosing numbers:
At round $i$, \Adam\ chooses a number $x_i \in \set{u_i, v_i}$
and \Eve\ chooses a number $y_i \in \set{w_i, z_i}$.
At the end of the game, the two players have jointly produced a sequence of numbers
$x_1, y_1, \dots, x_n, y_n$,
and \Eve\ wins the game if $x_1 + y_1 + \cdots + x_n + y_n = s$.

Given a \SSG as above, we construct a 1-clock \TBPP and a target configuration $\alpha$
s.t.~\Eve\ wins the \SSG iff $w$ is reachable from the initial configuration.
\Eve's moves are mimicked by branching transitions,
and \Adam's by nondeterministic choice.
An additional process checks that the total time elapsed in every branch is equal to the target sum $s$.

We can implement the above intuition in a \TBPP with a single clock $x$,
nonterminals $\?X=\{S,T, F\}\cup\{\forall_i,\exists_i,U_i,V_i\mid 1\le i \le n\} \cup \set{\exists_{n+1}}$ and rules as follows.
We have an initial production
\begin{align*}
    S \Rule{x=0} \forall _1\msplus T
\end{align*}
and a final rule that allows $T$ to check that the total elapsed time is $s$:
\begin{align*}
    T \Rule{x = s; x:=0} F
\end{align*}
For each round $1 \leq i \leq n$, the $i$-th move of \Eve\ is modelled by rules
\begin{align*}
    \exists_i \Rule{x=w_i;x:=0} \forall_i \quad\text{ and }\quad \exists_i \Rule{x=z_i; x=0} \forall_i
\end{align*}
and the $i$-th move of \Adam\ by
\begin{align*}
    \forall_i \Rule{} U_i \msplus V_i,
    \quad U_i \Rule{x=u_i;x:=0} \exists_{i+1},
    \quad\text{ and }\quad V_i \goesto {x=v_i; x=0} \exists_{i+1}.
\end{align*}
An extra rule concludes the simulation: $\exists_{i+1} \Rule{x = 0} F$.
The target multiset is $\alpha = (2^n + 1) \cdot F$, i.e., $2^n+1$ copies of $F$.
The correctness of the construction is stated in the lemma below.

\begin{lemma}
    \Eve\ wins the \SSG if, and only if, $(S, \vec{0}) \reach{*} (\alpha, \vec{0})$.
\end{lemma}

The reduction above produces a 1-\TBPP where
$(S, 0) \reach{*}(\alpha, \vec{0})+\gamma$ implies that $\gamma = \zerro$.
We thus get the following lower bound.
\begin{theorem}
The reachability and coverability problems for 1-clock \TBPP with target configuration encoded in binary are \PSPACE-hard.
\end{theorem}

\section{A proof of Lemma~\ref{lem:rasmus}}
\label{sec:a_proof_of_lem}

\newcommand{\blue}[1]{\textcolor{blue}{#1}}
\renewcommand{\red}[1]{\textcolor{red}{#1}}
\newcommand{\change}[2]{\red{#1}\blue{#2}}

\begin{definition}
 An \ELA formula over variables $x,y$ is \emph{basic} if it is $\true$, $\false$ or of form
$$\phi(x,y)\equiv (a \le_1 x \le_2 b) \land (c+dx\le_3 y)$$
for some constants $a,b,c\in\nnreals$, $d\in\{-1,0\}$ and $\le_1,\le_2,\le_3~\in\{<,\le\}$.
That is, its support is the upward-closure (wrt.~$y$) of a finite line segment with end points $a,b$, and slope $0$ or $-1$.
It is \emph{crossing} 
another basic formula $\phi'(x,y)\equiv (a' \le_1' x' \le_2' b') \land (c'+d'x'\le_3' y')$ 
at point $e\in\R$ if 
$(c+de = c'+d'e) \land (a \le_1 e \le_2 b) \land (a' \le_1' e \le_2' b')\land (d\neq d')$, i.e., their two line segments intersect.
A family $\?F$ of basic formulae is called \emph{simple} if no two elements cross.
\end{definition}

\begin{proposition}
    \label{prop:pre-basic}
    If $\phi(x,y)$ is a basic formula 
    then both
    $\varphi(x,y)\eqdef \exists t. \phi(x+t,y-t)$
    and
    $\varphi(x,y)\eqdef \exists t. \phi(x+t,y)$
    are basic formulae.
\end{proposition}

\begin{definition}
    Let $\phi$ be a formula and $\?F$ be a family of formulae over variables $x,y$.
We call $\phi$ \emph{piecewise} in $\?F$ if it is equivalent to a finite disjunction of formulae in $\?F$.
It is called \emph{piecewise basic} if it is the finite disjunction of basic formulae.

A family $\?F'$ a \emph{simplification} of $\?F$ if it is simple and every formula piecewise in $\?F$ is also piecewise in $\?F'$.
\end{definition}
\begin{proposition}\label{prop:simple_fun}
    Suppose $\?F$ is a simple family and $\varphi, \psi$ are piecewise in $\?F$.
    Then $\varphi\land \psi$ and $\varphi\lor \psi$ are piecewise in $\?F$.
\end{proposition}
\begin{proposition}
    \label{prop:simplification:computable}
    Every finite family $\?F$ of basic formulae has a unique (subset) minimal simplification, which can be computed in polynomial time.
\end{proposition}

Towards proving \cref{lem:rasmus},
recall that $\VAN{X}(x,y)$ holds if one can rewrite $(X,x)
\reach{y}\zerro$.
We will characterize these predicates using a type of game.
These are essentially one-dimensional priced timed games of \cite{BLMR2006,HIM2013},
with the provisio that Maximizer can only make discrete steps and cost rates can only be $0$ or $1$.
Additionally, we are interested in identifying if \emph{optimal} (value achieving) minimizer strategies exists, which is why we are dealing with value \emph{formulae} instead of simply value functions.

\newcommand{\tend}{T_f}
\begin{definition}\label{def:pgames}
    A \emph{\pgame} is given by a timed automaton $\?A=(\clockset,\?X, \?R)$,
    a partitioning $\?X=\?X_{min}\cup\?X_{\max}$ of the set of states into those belonging to players Min and Max, respectively,
    and a \emph{cost} $\cost(X) \in\{0,1\}$ for every state $X\in\?X$.
    Additionally, a game is parametrized by a payoff predicate $\payoff_X(x,y)$ for every state $X$.

A \emph{play} of the game is a walk $(s_0,t_0),(s_1,t_1),(s_2,t_2),\ldots$
in the configuration graph of $\?A$, where the player who owns the current state chooses the successor, with the restriction that only Min can use time steps.

The \emph{outcome} of infinite plays is $\infty$.
At any point in the play where the current configuration $(s_i,t_i) \in \?X_{\min}\x\R$ belongs to Min and
$\payoff_{s_i}(t_i,p)$ holds for some $p\in\R$,
she can stop the play and get the finite outcome $p + \sum_{i=0}^{n-1} \cost(s_i)\cdot(t_{i+1}-t_i)$.
The aim of players Min/Max is to minimize and maximize the outcome, respectively.

The \emph{outcome formulae} are the family $\{\varphi_X\mid X\in\?X\}$ of formulae such that 
$\varphi_X(x,y)$ holds iff Min has a strategy to guarantee a payoff of at most $y$ from initial configuration $(X,x)$.
\end{definition}

In order to compute the $\VAN{X}$ predicates for a given \TBPP we can equivalently compute the outcome formulae for a corresponding \pgame,
where Min gets to pick derivation rules and Max gets to pick the successor whenever a branching rule is used.
Every nonterminal owned by Min has cost rate $1$, and all other have rate $0$.
The payoff formula for $X\in\?X$ is simply $\payoff_X(x,y)\equiv \bigvee\{\varphi(x) \mid (X\Rule{\varphi; R}\zerro)\in\?R\}$, i.e., it is true iff $(X,x)\step{0}\zerro$, the configuration $(X,x)$ can vanish in no time.

We will solve these games iteratively, based on the special case for a single clock interval.

\begin{definition}
    Let $a,b\in\R$.
    An $(a,b)$-game 
    is a \pgame\ $\?G=(\?C,\?X,\?R,\cost,\payoff)$
    such that
            every rule $(X\Rule{\varphi;R}\alpha)\in\?R$ has the same clock constraint $\varphi(x)\equiv a<x<b$
            and does not reset the clock, and
            all payoff predicates are basic and with domain $(a,b)$.
    
    Let $\?F$ be the family of formulae that contains
    for every $\payoff_X(x,y)\equiv (\varphi(x)\land y\diamond c)$,
    the two formulae
    $(\varphi(x)\land y \diamond c)$
    and
    $(\varphi(x)\land y \diamond c-(b-x))$.
    Notice (\cref{prop:pre-basic}) that $\?F$ contains only basic formulae,
    and that there are at most $\card{\?X}$ points at which two formulae from $\?F$ can cross. 
    We write $\?F(\?G)$ for the simplification of $\?F$.
\end{definition}
The following properties of $\?F(\?G)$ follow from easy geometric reasoning.
\begin{proposition}
    \label{prop:FG-small}
    \begin{itemize}
        \item $\?F(\?G)$ has at fewer than $4\card{\?X}^3$ many elements. 
        \item Every formula piecewise in $\?F(\?G)$ is a disjunction of at most $2\card{\?X}(\card{\?X}-1)+1$ many basic formulae.
        \item If a formula $f$ is piecewise in $\?F(\?G)$ then so is $\bigvee_{0\le t \le (b-x)}(f(x+t,y-\cost(X)\cdot t)$.
    \end{itemize}
\end{proposition}
\begin{lemma}
    \label{lem:ab-game}
    The outcome formulae for the $(a,b)$-game $\?G$
    are piecewise in $\?F(\?G)$ and computable in polynomial time.
\end{lemma}
\begin{proof}
First, without loss of generality we may assume that there 
are no cycles in the game restricted only to the Max player states. If they are then their outcome must be $\outcome_X(x,y)\equiv\false$ as Max can enforce an infinite play.
Further, we assume that every transition going from a Max configuration is going to a Min 
player configuration, we simply introduce shortcuts for Max player.

Using the natural partial ordering $f\le g$ iff $f(x,y)\implies g(x,y)$ of two-variable formulae, we notice that
the outcome formulas $\{g_X\mid X\in\?X\}$ must be the least family of formulae satisfy the following equations.

\begin{equation}
    g_{X}(x,y)
    \quad=\quad
    \bigwedge_{X\Rule{\varphi}Y} g_{Y}(x,y)
    \qquad\text{for $X\in \?X_{\max}$, and}
\end{equation}
\begin{equation}
    \begin{aligned}
     g_{X}(x,y)
     \quad=\quad
     \payoff_X(x,y)
     &\quad\lor
     \bigvee_{(X\Rule{\varphi} Y)\in \?R} g_{Y}(x,y)\\
     &\quad\lor
     \bigvee_{0<t\le(b-x)}
     g_{X}(x+t,y-\cost(X)\cdot t)
 \end{aligned}
\end{equation}
for $X\in \?X_{min}$.
This allows to approximate them as the least fixed point of ordinal indexed families
$\?F^{(\alpha)}$,
where 
$\?F^{(0)}\eqdef \{g_X^{(0)} \eqdef \payoff_X \mid X\in\?X_{\min}\}\cup\{g_X^{(0)} \eqdef \false \mid X\in\?X_{\max}\}$, and for all $\alpha>0$
define 
\begin{equation}
    \label{eq:app1}
    g^{(\alpha)}_{X}(x,y)
    \quad\eqdef\quad
     \bigwedge_{X\Rule{\varphi}Y} g^{(\alpha)}_{Y}(x,y)
    \qquad\text{for $X\in \?X_{max}$, and}
\end{equation}
\begin{equation}
    \label{eq:app2}
    \bigvee_{\beta<\alpha}
    \left(
    g^{(\beta)}_X(x,y)
    \quad\lor
    \bigvee_{X\Rule{\varphi} Y} g^{(\beta)}_{Y}(x,y)
     \quad\lor
     \bigvee_{0<t\le(b-x)}
     (g^{(\beta)}_{X}(x+t,y-\cost(X)\cdot t))
    \right)
    \quad
\end{equation}
for $X\in \?X_{\min}$.
Notice that here, $g_Z^{(\alpha)} \le g_X^{(\beta)}$ for $\alpha<\beta$.
Intuitively,
$g^{(\alpha)}_X(x,y)$ holds if in a game starting in $(X,x)$, Min only needs to make $\alpha$ moves to get a payoff of at least $y$.

\begin{claim}
    \label{claim:finite-piecewise}
Every family $\?F^{(k)}$ with finite index $k\in\N$ is piecewise in $\?F(\?G)$.
\end{claim}
We prove this by induction on $k$,
where the base case, $\?F^{(0)}=\payoff$ trivially holds.
For the induction step, consider the case for $X\in\?X_{\min}$ 
and let
\begin{equation}
    \begin{aligned}
        g^{(k+1)}_{X}(x,y)
        \quad =\quad&
        g^{(k)}_{X}(x,y)~\lor\\
                           &\left(
    \bigvee_{X\Rule{\varphi} Y} g^{(k)}_{Y}(x,y)
     \quad\lor
     \bigvee_{0<t\le(b-x)}
     (g^{(k)}_{X}(x+t,y-\cost(X)\cdot t))
     \right)
    \end{aligned}
\end{equation}
Here, $(g^{(k)}_{X}(x+t,y-\cost(X)\cdot t))$
is piecewise in $\?F(\?G)$ by
(the last point in) \cref{prop:FG-small},
because $g^{(k)}_{X}$ is piecewise in $\?F(\?G)$ by induction assumption.
Together with \cref{prop:simple_fun} this shows that $g_X^{(k+1)}$ must be piecewise in $\?F(\?G)$.
The case for Maximizer states is analogous.

\begin{claim}
    \label{claim:finite-stabilization}
    There exists $k \le 4\card{\?X}^3$ such that $\?F^{(k+1)} = \?F^{(k)}$.
\end{claim}
This follows from \cref{claim:finite-piecewise} and \cref{prop:FG-small}:
Every time $\?F^{(k)}$ strictly improves
there must be at least one of its elements $g^{(k)}_X(x,y) = \bigvee_{i\le \card{\?X}} f_i$
where at least one disjunct $f_i$ is replaced by a strictly larger one.
This can happen at most $4\card{\?X}^3$ times.

It remains to observe that one can compute
a representation of $\?F^{(k+1)}$ from $\?F^{(k)}$, using \cref{eq:app1,eq:app2}.
\end{proof}

\Cref{lem:rasmus} now rests on the following lemma, which we prove by recursively applying \cref{lem:ab-game} to solve \pgame s without resets and transition guards.
\begin{lemma}
    Let $(\clockset,\?X, \?R,\cost,\payoff)$ be a \pgame\
    where $\payoff$ is piecewise basic.
    Let $k\in\N$ denote the number of different constants appearing either as endpoints in $\payoff$ or in the guards of rules.
    Every outcome formula $\outcome_X(x,y)$ is a disjunction of at most $2\cdot k\cdot\card{\?X}$ basic formulae
    and is computable in polynomial time.
\end{lemma}
\begin{proof}
    We can partition the reals into $2 k+1$ intervals according to the constraint constants
    $$\{a_0=0\}, (a_0,a_1), \{a_1\}, (a_1,a_2) \{a_2\},\ldots (a_{k-1},a_{k}=\infty)$$
    and write $\outcome_X(x,y)\equiv \bigvee_{i\le k} f_X^i$, where
    $f^{2i}_X(x,y)\eqdef\outcome_X(x,y)\land (x=a_i)$ for even indices
    $f^{2i+1}_X(x,y)\eqdef\outcome_X(x,y)\land (a_i<x<a_{i+1})$ for odd indices.

    Suppose first that no rule in the \pgame\ resets a clock to $0$.
    In this case $f_X^i$ are independent of $f_X^j$ with smaller index $j<i$ and can be computed stepwise from last to first interval,
    preserving the invariant that each $f_X^i$ is 
    the disjunction of at most $\card{\?X}$ many basic formulae.
    \begin{itemize}
        \item For the final interval we must have $f_X^{2k+1}(x,y)\equiv (x>a_k)\land \payoff_X(x,y)$.
    
        \item For even indices (singleton intervals $\{a_i\}$), 
            outcome formulae
            can be computed once $\{f^{2i+1}_Y\mid Y\in\?X\}$ are known, using dynamic programming.
            Briefly, compute payoff formulae of the form
            $\phi_X(x,y)\equiv (x=a_i)\land y\diamond c$, for $\diamond\in\{>,\ge\}$ and $c\in\R$,
            from the $f^{2i+1}$, which is possible because those are piecewise basic and by \cref{prop:pre-basic},
            and solve an untimed min/max game for those payoffs.
        \item For odd indices corresponding to intervals $(a_i,a_{i+1})$, we can use
            \cref{lem:ab-game} to 
            compute $f^{2i+1}_X$ as the outcome formulae
            for an $(a_i,a_{i+1})$-game with payoffs given by
          $\payoff_X(x,y) \equiv (x=a_{i+1})\land f^i_X(x,y)$.
\end{itemize}
This produces outcome formulae which are piecewise basic and expressed as disjunctions of at most $2k\card{\?X}+1$ basic formulae.

It remains to argue that the above construction can be extended to \pgame s \emph{with} resetting rules.
This is based on the observation that Minimizer can be required not to allow a configuration to repeat along a play. 
Indeed, the outcome of finite a play
$$(s_0,t_0),\ldots,(s_0,t_0),\ldots,(s_k,t_k)p$$
where $p\in\R$ as chosen by Min at the end of a finite play,
can only be larger than that of the suffix $(s_0,t_0),\ldots,(s_k,t_k)p$ alone. Any Min strategy that allows for repetitions can therefore be turned into one that does not,
by cutting out the intermediate paths.
Consequently, plays with more than $\card{\?X}$ clock resets can be declared to be losing (have outcome $\infty$)
without changing the resulting outcome formulae.

In order to compute outcome formulae for an unrestricted \pgame~we can now
stepwise compute formulae $\outcome^{(i)}_X$ for the same game but where a play ends (with outcome $\infty$) as soon as Min uses the $i$th reset.
By the observation above, $\outcome = \outcome^{\card{\?X}}$.
The base case $\outcome^{(0)}_X$ corresponds to the game where all reset rules are removed.

In order to compute
$\outcome^{(i+1)}_X$ we can use the same procedure as described above,
but in every $(a,b)$-game,
we replace resetting rules
$X\Rule{\varphi, x:=0}Y$
by non-resetting rules $X\Rule{\varphi}Y'$
to new, unproductive nonterminals $\hat{Y}$
with
$\payoff_{\hat{Y}}(x,y)\eqdef (x=b) \land \outcome^{(i)}_Y(0,y)$,
and $\cost(\hat{Y}) =0$.
This way, if from $(X,t)$, Min would originally gain an outcome by resetting to $(Y,0)$,
she can now gain the same outcome by moving to $(Y',t)$, delaying by $(b-t)$ at no cost,
and collect the outcome in configuration $(\hat{Y},b)$.
This change will not introduce any new constants in the guards of rules, so the total number of intervals 
in the description of the outcome formulae remains unchanged.
\end{proof}

\end{document}